\documentclass[a4paper,UKenglish,cleveref, autoref, thm-restate]{lipics-v2021}

\usepackage{breakurl}
\usepackage{microtype}
\usepackage{hyperref}
\usepackage{amsmath}
\usepackage{amsfonts}
\usepackage{amssymb}
\usepackage{amsthm}
\usepackage{mathtools}
\usepackage{MnSymbol}
\usepackage{moreverb}
\usepackage{cite}
\usepackage{braket}
\usepackage{stmaryrd}
\usepackage{cmll}
\usepackage{array}
\usepackage{algorithm2e}
\usepackage{bussproofs}
\usepackage{tikz-cd}
\usepackage{xspace}
\usepackage[draft]{tikzit}
\input{blocks.tikzdefs}

\tikzstyle{box}=[fill=white, draw=black, shape=rectangle, minimum width=5mm, minimum height=5mm]
\tikzstyle{dot}=[fill=black, draw=black, shape=circle]
\tikzstyle{discard}=[circuit ee IEC, thick, ground, rotate=90, scale=1.5, inner sep=-2mm, tikzit fill={rgb,255: red,176; green,176; blue,176}]
\tikzstyle{maxmix}=[circuit ee IEC, thick, ground, rotate=-90, scale=1.5, inner sep=-2mm, tikzit fill={rgb,255: red,102; green,102; blue,102}]

\tikzstyle{arr}=[->]
\tikzstyle{shade}=[-, fill={rgb,255: red,191; green,191; blue,191}, opacity=.5]
\tikzstyle{empty}=[-, fill=white]
\tikzstyle{dashes}=[-, dashed]

\allowdisplaybreaks

\hideLIPIcs
\nolinenumbers

\newcommand{\MatRp}{\ensuremath{\textrm{Mat}[\mathbb R^+]}\xspace}
\newcommand{\MatRaff}{\ensuremath{\textrm{Mat}[\mathbb R]}\xspace}
\newcommand{\CP}{\ensuremath{\textrm{CP}}\xspace}
\newcommand{\CPs}{\ensuremath{\textrm{CP}^*}\xspace}

\newcommand{\id}{\mathrm{id}}
\newcommand{\bv}{\mathrm{BV}}
\newcommand{\caus}{\mathrm{Caus}}
\newcommand{\sub}{\mathrm{Sub}}
\newcommand{\catc}{\mathcal{C}}
\newcommand{\ob}{\mathrm{Ob}}
\newcommand{\fsub}[1]{[#1]}

\makeatletter
\def\namedlabel#1#2{\begingroup
    #2%
    \def\@currentlabel{#2}%
    \phantomsection\label{#1}\endgroup
}
\makeatother

\newcommand\discard{
\mathbin{\text{\begin{tikzpicture}[circuit ee IEC,yscale=0.6,xscale=0.5]
\draw (0,-2ex) to (0,0) node[ground,rotate=90,xshift=.65ex] {};
\end{tikzpicture}}}
}
\newcommand\maxmix{
\mathbin{\text{\begin{tikzpicture}[circuit ee IEC,yscale=0.6,xscale=0.5]
\draw (0,2ex) to (0,0) node[ground,rotate=-90,xshift=.65ex] {};
\end{tikzpicture}}}
}

\newcommand{\todoo}[1]{}

\newcommand{\aff}{\textsf{\sf aff}}
\newcommand{\affp}{\textsf{\sf aff}^+}

\title{Higher-order causal theories are models of $\bv$-logic}

\author{Will Simmons}{Department of Comptuter Science, University of Oxford, Oxford, UK \and Cambridge Quantum, Terrington House, 13-15 Hills Road, Cambridge, UK}{will.simmons@cambridgequantum.com}{https://orcid.org/0000-0001-7996-6577}{}

\author{Aleks Kissinger}{Department of Comptuter Science, University of Oxford, Oxford, UK}{aleks.kissinger@cs.ox.ac.uk}{https://orcid.org/0000-0002-6090-9684}{}

\authorrunning{W. Simmons and A. Kissinger}
\Copyright{Will Simmons and Aleks Kissinger}

\ccsdesc[500]{Theory of computation~Quantum computation theory}
\ccsdesc[500]{Theory of computation~Linear logic}
\ccsdesc[300]{Theory of computation~Categorical semantics}

\keywords{Causality, linear logic, categorical logic, probabilistic coherence spaces, quantum channels}

\category{} 
\relatedversion{} 

\acknowledgements{}

\begin{document}

\maketitle

\begin{abstract}
The $\caus[-]$ construction takes a compact closed category of basic processes and yields a *-autonomous category of higher-order processes obeying certain signalling/causality constraints, as dictated by the type system in the resulting category. This paper looks at instances where the base category C satisfies additional properties  yielding an affine-linear structure on $\caus[\catc]$ and a substantially richer internal logic. While the original construction only gave multiplicative linear logic, here we additionally obtain additives and a non-commutative, self-dual sequential product yielding a model of Guglielmi's BV logic. Furthermore, we obtain a natural interpretation for the sequential product as ``A can signal to B, but not vice-versa'', which sits as expected between the non-signalling tensor and the fully-signalling (i.e. unconstrained) par. Fixing matrices of positive numbers for $\catc$ recovers the BV category structure of probabilistic coherence spaces identified by Blute, Panangaden, and Slavnov, restricted to normalised maps. On the other hand, fixing the category of completely positive maps gives an entirely new model of BV consisting of higher order quantum channels, encompassing recent work in the study of quantum and indefinite causal structures.
\end{abstract}

\section{Introduction}\label{sec:introduction}

The causality condition \cite{Chiribella2010} is a simple equation one can impose on a family of processes that states essentially that ``discarding the output of $f$ is the same as discarding its input.'' Pictorially:\footnote{We will use \textit{string diagram} notation for monoidal categories throughout, depicting morphisms as boxes, composition as ``plugging'' boxes from bottom-to-top and $\otimes$ as placing boxes side-by-side. See e.g.~\cite{selinger2009survey}.}
\begin{equation}\label{eq:causality}
  \hfill\begin{tikzpicture}
	\begin{pgfonlayer}{nodelayer}
		\node [style=box] (0) at (-2, 0) {$f$};
		\node [style=none] (1) at (-2, -1.5) {};
		\node [style=discard] (2) at (-2, 1.75) {};
		\node [style=none] (3) at (0, 0) {$=$};
		\node [style=none] (4) at (-1.5, -1) {$A$};
		\node [style=none] (5) at (-1.5, 1) {$B$};
		\node [style=none] (6) at (1.5, -1.5) {};
		\node [style=none] (7) at (2, -1) {$A$};
		\node [style=discard] (8) at (1.5, 0.5) {};
	\end{pgfonlayer}
	\begin{pgfonlayer}{edgelayer}
		\draw (1.center) to (0);
		\draw (0) to (2);
		\draw (6.center) to (8);
	\end{pgfonlayer}
\end{tikzpicture}\hfill
\end{equation}
for some fixed family of ``discarding'' processes $\discard_A$. While this seems to be a simple equation, and reasonable to impose on most sane collections of physical processes, it deserves the somewhat lofty name of \textit{causality} because it really seems to capture the essence of influences moving in a single direction: forward in time. In particular, it guarantees that distant agents cannot send messages to one another without their existing a forward-directed path in the diagram of their processes. Without referring explicitly to spacetime, this serves as an abstract stand-in for what a relativity theorist would also call causality, namely the limitation that influences cannot propagate faster than the speed of light~\cite{kissinger2017equivalence}.

In concrete categories of interest, namely those of probabilistic and quantum maps, equation \eqref{eq:causality} simply imposes that processes should preserve normalised states. Iterating this, we might wish to consider second-order processes that preserve processes that preserve normalised states, and so on. This turns out the provide a rich landscape for modelling causal relationships between events, enabling for example chains~\cite{chiribella2008quantum,Gutoski2007} (or directed acyclic graphs~\cite{Kissinger2019a}) of events in a definite causal order, probabilistic mixtures of causal orderings, or even exotic indefinite causal structures~\cite{Chiribella2013,Oreshkov2012,baumeler2014maximal}.


There are several routes to constructing a framework for characterising and composing such higher-order processes, including describing acyclic graphs of interactions between discrete first-order channels \cite{Blute2003,Chiribella2009} or building an infinite hierarchy of types recursively out of constructors describing a kind of causal relationship \cite{Bisio2019}.

The $\caus[-]$ construction of Kissinger and Uijlen \cite{Kissinger2019a} fits the latter style through a category where morphisms map normalised states of the source type to normalised states of the target type. The typing of a morphism is a judgement about externally-observable properties of normalisation and information flow as opposed to any assumed internal structure or spacetime geometry. The type constructors model the operators of Multiplicative Linear Logic (MLL), i.e. it forms a $*$-autonomous category. MLL features two logical connectives, tensor and par, which serve as two extremes in the $\caus[-]$ construction. Tensor yields a \textit{non-signalling} composition of processes, where causal influences are not allowed to pass from one side to the other, whereas par gives a \textit{fully-signalling} composition, i.e. one that imposes no signalling constraints.

In between this lies the \textit{one-way signalling} processes, where causal influences can flow from one agent (say, Alice in the past) to another (say, Bob in the future). While special cases of such processes were treated in an \textit{ad hoc} way in \cite{Kissinger2019a}, here we will show that one-way signalling can in fact be treated as a fully-fledged connective in its own right, yielding a substantially richer logical structure.

$\bv$-logic \cite{Guglielmi2007} adds a third logical connective that is non-commutative to capture sequentiality. It admits a categorical characterisation via $\bv$-categories, and has previously been applied to the study of (probabilistic) coherence spaces \cite{Blute2012} and a certain graph-based model of quantum causal structures called discrete quantum causal dynamics~\cite{Blute2003}.


In this paper, we will adapt the $\caus[-]$ construction by modifying some of the assumptions on the base category, requiring it to be \textit{additive precausal}. The extra structure allows us to consider new ways of combining processes and types to describe operational constructions such as binary tests, probabiity distributions, and one-way signalling processes. These correspond to extending the logical structure of $\caus[\catc]$ with the additive connectives (sometimes called ``with'' and ``plus'') of linear logic, as well as the sequential connective of $\bv$. We show that the main classical and quantum examples of precausal categories are furthermore additive precausal, enabling $\caus[\catc]$ to model higher-order theories and classical and quantum causal structures, respectively. We also note the existence of non-standard models, such as higher-order affine (a.k.a. quasi-probabilistic) processes.

\section{Higher-Order Causal Structure}\label{sec:caus}

We use ``process theory'' style terminology (see e.g.~\cite{CKbook}) for monoidal categories throughout. Namely, we use the terms \textit{type} or \textit{system} interchangeably to refer to objects of a symmetric monoidal category, \textit{process} to refer to a generic morphism $f : A \to B$, \textit{state} to refer to a morphism $\rho : I \to A$ from the tensor unit, \textit{effect} to refer to a morphism $\pi : A \to I$ to the tensor unit, and \textit{scalar} to refer to a morphism $\lambda : I \to I$. In categories admitting higher-order structure, we think of first-order types as state spaces and first-order processes as maps that transform first-order states to first-order states. Higher-order processes are transformations where the input or output system is itself a map. The main example used in causality literature is a quantum 2-comb which is a process taking a channel over first-order types as an input and transforms it into a new channel, typically by composing with some pre- and post-processing which may share some memory channel. One can extend this to an infinite hierarchy of higher-order processes where an $(n+1)$-comb transforms an $n$-comb into a $1$-comb (a channel) \cite{Chiribella2009}. A higher-order process theory is one which describes processes in such an infinite hierarchy uniformly.


Higher-order theories are commonly achieved by providing a mechanism to encode transformations into states of some function type. In category theory, this corresponds to an internal hom in monoidal closed categories, i.e. a bifunctor $\multimap : \catc^{op} \times \catc \to \catc$ such that $\catc(C \otimes A, B) \simeq \catc(C, A \multimap B)$ and this is natural in all arguments. Wilson and Chiribella \cite{Wilson2021a} demonstrate that adding basic manipulations (morphisms that capture sequential and parallel composition of such encoded functions) is sufficient to permit the inductive generation of comb types. However, with higher-order theories, there are more ways of composing processes that need not be obtainable in this way.


To solve this, we may further ask that $\catc$ be a \textit{compact closed category}, i.e. we require every object $A$ has a ``cup'' state $\eta_A : I \to A^* \otimes A$ and a ``cap'' effect satisfying the so-called ``yanking equations'': $(\id_A \otimes \eta_A) \fatsemi (\epsilon_A \otimes \id_A) = \id_A$ and $(\eta_A \otimes \id_{A^*}) \fatsemi (\id_{A^*} \otimes \epsilon_A) = \id_{A^*}$. These enable us to convert inputs to outputs at will, which in turn lets us ``wire up'' processes in arbitrary ways to each other. Namely, we can treat all higher order processes as states, and use ``caps'' to connect them to each other in arbitrary ways. See Figure \ref{fig:compact_comb} for an example and \cite[Section 2.1]{Kissinger2019a} for more details.



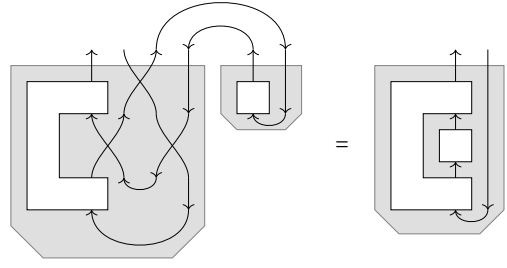
\begin{figure}
\centering
\begin{equation*}
\hfill\begin{tikzpicture}
	\begin{pgfonlayer}{nodelayer}
		\node [style=none] (0) at (1.75, -2) {};
		\node [style=none] (1) at (-0.75, -2) {};
		\node [style=none] (2) at (-3.25, -2) {};
		\node [style=none] (3) at (-0.75, -1) {};
		\node [style=none] (4) at (-2.25, -1) {};
		\node [style=none] (5) at (-2.25, 1) {};
		\node [style=none] (6) at (-0.75, 1) {};
		\node [style=none] (7) at (-0.75, 2) {};
		\node [style=none] (8) at (-3.25, 2) {};
		\node [style=none] (9) at (-1.25, -2) {};
		\node [style=none] (10) at (-1.25, -1) {};
		\node [style=none] (11) at (-0.25, 1) {};
		\node [style=none] (12) at (-0.25, -1) {};
		\node [style=none] (13) at (-1.25, 1) {};
		\node [style=none] (14) at (-1.25, 2) {};
		\node [style=none] (15) at (-1.25, 3) {};
		\node [style=none] (16) at (0.75, -1) {};
		\node [style=none] (17) at (1.75, -1) {};
		\node [style=none] (18) at (0.75, 1) {};
		\node [style=none] (19) at (1.75, 1) {};
		\node [style=none] (20) at (1.75, 3) {};
		\node [style=none] (21) at (0.75, 3) {};
		\node [style=none] (22) at (-0.25, 3) {};
		\node [style=none] (23) at (2.25, 2.5) {};
		\node [style=none] (24) at (-3.75, 2.5) {};
		\node [style=none] (25) at (-3.75, -2.5) {};
		\node [style=none] (26) at (2.25, -2.5) {};
		\node [style=none] (27) at (1.25, -3.5) {};
		\node [style=none] (28) at (-2.75, -3.5) {};
		\node [style=none] (29) at (4.25, 1) {};
		\node [style=none] (30) at (4.25, 2) {};
		\node [style=none] (31) at (3.25, 2) {};
		\node [style=none] (32) at (3.25, 1) {};
		\node [style=none] (33) at (3.75, 2) {};
		\node [style=none] (34) at (3.75, 1) {};
		\node [style=none] (35) at (4.75, 1) {};
		\node [style=none] (36) at (4.75, 3) {};
		\node [style=none] (37) at (3.75, 3) {};
		\node [style=none] (38) at (5.25, 2.5) {};
		\node [style=none] (39) at (2.75, 2.5) {};
		\node [style=none] (40) at (2.75, 1) {};
		\node [style=none] (41) at (3.25, 0.5) {};
		\node [style=none] (42) at (4.75, 0.5) {};
		\node [style=none] (43) at (5.25, 1) {};
	\end{pgfonlayer}
	\begin{pgfonlayer}{edgelayer}
		\draw [style=shade] (26.center)
			 to (23.center)
			 to (24.center)
			 to (25.center)
			 to (28.center)
			 to (27.center)
			 to cycle;
		\draw [style=shade] (39.center)
			 to (40.center)
			 to (41.center)
			 to (42.center)
			 to (43.center)
			 to (38.center)
			 to cycle;
		\draw [style=empty] (2.center)
			 to (1.center)
			 to (3.center)
			 to (4.center)
			 to (5.center)
			 to (6.center)
			 to (7.center)
			 to (8.center)
			 to cycle;
		\draw [style=arr] (14.center) to (15.center);
		\draw [style=arr, in=-90, out=90, looseness=0.75] (10.center) to (11.center);
		\draw [style=arr, in=-90, out=90, looseness=0.75] (12.center) to (13.center);
		\draw [style=arr, bend left=90, looseness=1.25] (16.center) to (12.center);
		\draw [style=arr, bend left=90, looseness=1.25] (0.center) to (9.center);
		\draw [style=arr] (17.center) to (0.center);
		\draw [in=-90, out=90, looseness=0.75] (17.center) to (18.center);
		\draw [in=-90, out=90, looseness=0.75] (18.center) to (22.center);
		\draw [style=arr, in=90, out=-90, looseness=0.75] (19.center) to (16.center);
		\draw [style=arr, in=-90, out=90, looseness=0.75] (11.center) to (21.center);
		\draw [style=arr] (20.center) to (19.center);
		\draw [style=empty] (29.center)
			 to (30.center)
			 to (31.center)
			 to (32.center)
			 to cycle;
		\draw [style=arr] (33.center) to (37.center);
		\draw [style=arr, bend right=270, looseness=1.25] (35.center) to (34.center);
		\draw [style=arr] (36.center) to (35.center);
		\draw [style=arr, bend right=90, looseness=1.25] (37.center) to (20.center);
		\draw [style=arr, bend left=90, looseness=1.25] (21.center) to (36.center);
	\end{pgfonlayer}
\end{tikzpicture} \ \ = \ \ \begin{tikzpicture}
	\begin{pgfonlayer}{nodelayer}
		\node [style=none] (0) at (2, -2) {};
		\node [style=none] (1) at (1.5, -2) {};
		\node [style=none] (2) at (-1, -2) {};
		\node [style=none] (3) at (1.5, -1) {};
		\node [style=none] (4) at (0, -1) {};
		\node [style=none] (5) at (0, 1) {};
		\node [style=none] (6) at (1.5, 1) {};
		\node [style=none] (7) at (1.5, 2) {};
		\node [style=none] (8) at (-1, 2) {};
		\node [style=none] (9) at (1, -2) {};
		\node [style=none] (10) at (1, -1) {};
		\node [style=none] (11) at (1, -0.5) {};
		\node [style=none] (13) at (1, 1) {};
		\node [style=none] (14) at (1, 2) {};
		\node [style=none] (15) at (1, 3) {};
		\node [style=none] (22) at (2, 3) {};
		\node [style=none] (23) at (2.5, 2.5) {};
		\node [style=none] (24) at (-1.5, 2.5) {};
		\node [style=none] (25) at (-1.5, -2) {};
		\node [style=none] (26) at (2.5, -2) {};
		\node [style=none] (27) at (1.75, -2.75) {};
		\node [style=none] (28) at (-0.75, -2.75) {};
		\node [style=none] (29) at (1.5, -0.5) {};
		\node [style=none] (30) at (1.5, 0.5) {};
		\node [style=none] (31) at (0.5, 0.5) {};
		\node [style=none] (32) at (0.5, -0.5) {};
		\node [style=none] (44) at (1, 0.5) {};
	\end{pgfonlayer}
	\begin{pgfonlayer}{edgelayer}
		\draw [style=shade] (26.center)
			 to (23.center)
			 to (24.center)
			 to (25.center)
			 to (28.center)
			 to (27.center)
			 to cycle;
		\draw [style=empty] (2.center)
			 to (1.center)
			 to (3.center)
			 to (4.center)
			 to (5.center)
			 to (6.center)
			 to (7.center)
			 to (8.center)
			 to cycle;
		\draw [style=arr] (14.center) to (15.center);
		\draw [style=arr, in=-90, out=90, looseness=0.75] (10.center) to (11.center);
		\draw [style=arr, bend left=90, looseness=1.25] (0.center) to (9.center);
		\draw [style=empty] (29.center)
			 to (30.center)
			 to (31.center)
			 to (32.center)
			 to cycle;
		\draw [style=arr] (44.center) to (13.center);
		\draw [style=arr] (22.center) to (0.center);
	\end{pgfonlayer}
\end{tikzpicture}\hfill
\end{equation*}
\caption{Example of embedding higher-order processes and composition via compact closure.}
\label{fig:compact_comb}
\end{figure}

The causality condition in equation~\eqref{eq:causality} 
specialises to states as the requirement that $\rho \fatsemi \discard_A = \id_I$, i.e. it imposes that states be normalised with respect to the discarding effect. Discarding is unique in the sense that it is the only effect normalised for all states. However, for the higher-order analogue of state spaces, there might be more such effects, motivating the following definition.

\begin{definition}\label{def:dual}
  For $c \subseteq \catc(I,A)$, the dual set is $c^* := \{ \pi \in \catc(A, I) \,|\, \forall \rho \in \rho \fatsemi \pi = \id_I \}$.
%
\end{definition}

\begin{remark}
In this paper, we will freely mix between treating $c$/$c^*$ as a set of states in $\catc(I, A)$/$\catc(I, A^*)$ and a set of effects in $\catc(A^*, I)$/$\catc(A, I)$, depending on the context. They are equivalent through the transpose isomorphism induced by compact closure.
\end{remark}

From this we can recover the space of first-order/\textit{causal} states for a system $A$ as $\left\{\discard_A\right\}^*$. For example, in quantum theory we take $\discard_A \in \catc(A, I)$ to represent the partial trace over the system $A$, so this set describes the space of trace-$1$ density matrices. First-order/causal processes from $A$ to $B$ are those that map every state in $\left\{\discard_A\right\}^*$ to a state in $\left\{\discard_B\right\}^*$ (equivalent to mapping the effect $\discard_B$ to $\discard_A$ when our category has enough states/well-pointedness), which we can encode as $\left\{\left\{\discard_A\right\}^* \otimes \discard_B\right\}^* \subseteq \catc(I, A^* \otimes B)$. These concepts are used to define a precausal category, which specifies one possible set of conditions to enable discussions of higher-order causal structures.

\begin{definition}\label{def:precausal}
A compact closed category $\catc$ is a \textit{precausal category} if:

\begin{enumerate}[PC1.]
  \item[{\rm \namedlabel{prec:discard}{PC1}.}] $\catc$ has a discarding process $\discard_A \in \catc(A, I)$ for every system $A$, compatible with the monoidal structure as below;
\begin{align}
\discard_{A \otimes B} &= \discard_A \otimes \discard_B \\
\discard_I &= \id_I
\end{align}
\item[{\rm \namedlabel{prec:dimension}{PC2}.}] The dimension $d_A := \maxmix_A \fatsemi \discard_A$ is invertible for all non-zero $A$;
\item[{\rm \namedlabel{prec:enough}{PC3}.}] $\catc$ has enough causal states: $\forall f, g : A \to B . \left(\forall \rho \in \left\{\discard_A\right\}^* . \rho \fatsemi f = \rho \fatsemi g\right) \Rightarrow f = g$;
\item[{\rm \namedlabel{prec:oneway}{PC4}.}] Causal one-way signalling processes on first-order types factorise: for any causal $\Phi : A \otimes B \to C \otimes D$,
\begin{equation}
\left(
\begin{array}{c}
\exists \Phi' : A \to C \text{ causal} . \\
\begin{tikzpicture}
	\begin{pgfonlayer}{nodelayer}
		\node [style=none] (0) at (-1, 0.5) {};
		\node [style=none] (1) at (1, 0.5) {};
		\node [style=none] (2) at (1, -0.5) {};
		\node [style=none] (3) at (-1, -0.5) {};
		\node [style=none] (4) at (0, 0) {$\Phi$};
		\node [style=none] (5) at (-0.5, 0.5) {};
		\node [style=none] (6) at (-0.5, 1) {};
		\node [style=discard] (7) at (0.5, 1) {};
		\node [style=none] (8) at (0.5, 0.5) {};
		\node [style=none] (9) at (-0.5, -0.5) {};
		\node [style=none] (10) at (-0.5, -1) {};
		\node [style=none] (11) at (0.5, -0.5) {};
		\node [style=none] (12) at (0.5, -1) {};
	\end{pgfonlayer}
	\begin{pgfonlayer}{edgelayer}
		\draw [style=empty] (1.center)
			 to (2.center)
			 to (3.center)
			 to (0.center)
			 to cycle;
		\draw (6.center) to (5.center);
		\draw (7) to (8.center);
		\draw (9.center) to (10.center);
		\draw (11.center) to (12.center);
	\end{pgfonlayer}
\end{tikzpicture} = \begin{tikzpicture}
	\begin{pgfonlayer}{nodelayer}
		\node [style=none] (0) at (-1, 0.5) {};
		\node [style=none] (1) at (0, 0.5) {};
		\node [style=none] (2) at (0, -0.5) {};
		\node [style=none] (3) at (-1, -0.5) {};
		\node [style=none] (4) at (-0.5, 0) {$\Phi'$};
		\node [style=none] (5) at (-0.5, 0.5) {};
		\node [style=none] (6) at (-0.5, 1) {};
		\node [style=discard] (7) at (0.5, -0.5) {};
		\node [style=none] (8) at (0.5, -1) {};
		\node [style=none] (9) at (-0.5, -0.5) {};
		\node [style=none] (10) at (-0.5, -1) {};
	\end{pgfonlayer}
	\begin{pgfonlayer}{edgelayer}
		\draw [style=empty] (1.center)
			 to (2.center)
			 to (3.center)
			 to (0.center)
			 to cycle;
		\draw (6.center) to (5.center);
		\draw (7) to (8.center);
		\draw (9.center) to (10.center);
	\end{pgfonlayer}
\end{tikzpicture}
\end{array}
\right) \Rightarrow \left(
\begin{array}{c}
\exists Z, \Phi_1 : A \to C \otimes Z \text{ causal}, \\
\Phi_2 : Z \otimes B \to D \text{ causal} . \\
\begin{tikzpicture}
	\begin{pgfonlayer}{nodelayer}
		\node [style=none] (0) at (-1, 0.5) {};
		\node [style=none] (1) at (1, 0.5) {};
		\node [style=none] (2) at (1, -0.5) {};
		\node [style=none] (3) at (-1, -0.5) {};
		\node [style=none] (4) at (0, 0) {$\Phi$};
		\node [style=none] (5) at (-0.5, 0.5) {};
		\node [style=none] (6) at (-0.5, 1) {};
		\node [style=none] (8) at (0.5, 0.5) {};
		\node [style=none] (9) at (-0.5, -0.5) {};
		\node [style=none] (10) at (-0.5, -1) {};
		\node [style=none] (11) at (0.5, -0.5) {};
		\node [style=none] (12) at (0.5, -1) {};
		\node [style=none] (13) at (0.5, 1) {};
	\end{pgfonlayer}
	\begin{pgfonlayer}{edgelayer}
		\draw [style=empty] (1.center)
			 to (2.center)
			 to (3.center)
			 to (0.center)
			 to cycle;
		\draw (6.center) to (5.center);
		\draw (9.center) to (10.center);
		\draw (11.center) to (12.center);
		\draw (13.center) to (8.center);
	\end{pgfonlayer}
\end{tikzpicture} = \begin{tikzpicture}
	\begin{pgfonlayer}{nodelayer}
		\node [style=none] (0) at (-1.5, -1.25) {};
		\node [style=none] (1) at (0.5, -1.25) {};
		\node [style=none] (2) at (0.5, -0.25) {};
		\node [style=none] (3) at (-1.5, -0.25) {};
		\node [style=none] (4) at (-0.5, -0.75) {$\Phi_1$};
		\node [style=none] (5) at (-0.5, 0.25) {};
		\node [style=none] (6) at (1.5, 0.25) {};
		\node [style=none] (7) at (1.5, 1.25) {};
		\node [style=none] (8) at (-0.5, 1.25) {};
		\node [style=none] (9) at (0.5, 0.75) {$\Phi_2$};
		\node [style=none] (10) at (1, 0.25) {};
		\node [style=none] (11) at (1, -1.75) {};
		\node [style=none] (12) at (1, 1.25) {};
		\node [style=none] (13) at (1, 1.75) {};
		\node [style=none] (14) at (-1, -0.25) {};
		\node [style=none] (15) at (-1, 1.75) {};
		\node [style=none] (16) at (-1, -1.25) {};
		\node [style=none] (17) at (-1, -1.75) {};
		\node [style=none] (18) at (0, -0.25) {};
		\node [style=none] (19) at (0, 0.25) {};
	\end{pgfonlayer}
	\begin{pgfonlayer}{edgelayer}
		\draw [style=empty] (2.center)
			 to (1.center)
			 to (0.center)
			 to (3.center)
			 to cycle;
		\draw [style=empty] (7.center)
			 to (6.center)
			 to (5.center)
			 to (8.center)
			 to cycle;
		\draw (17.center) to (16.center);
		\draw (14.center) to (15.center);
		\draw (18.center) to (19.center);
		\draw (12.center) to (13.center);
		\draw (11.center) to (10.center);
	\end{pgfonlayer}
\end{tikzpicture}
\end{array}
\right)
\end{equation}
\item[{\rm \namedlabel{prec:split}{PC5}.}] For all $w : I \to A \otimes B^*$:
\begin{equation}
\left(
\begin{array}{c}
\forall \Phi : A \to B \text{ causal} . \\
\begin{tikzpicture}
	\begin{pgfonlayer}{nodelayer}
		\node [style=none] (0) at (-1, -0.5) {};
		\node [style=none] (1) at (1, -0.5) {};
		\node [style=none] (2) at (-1, -1.5) {};
		\node [style=none] (3) at (1, -1.5) {};
		\node [style=none] (4) at (0, -1) {$w$};
		\node [style=none] (5) at (-0.5, -0.5) {};
		\node [style=none] (6) at (0.5, 0.25) {};
		\node [style=none] (7) at (-1, 1) {};
		\node [style=none] (8) at (0, 1) {};
		\node [style=none] (9) at (-1, 0) {};
		\node [style=none] (10) at (0, 0) {};
		\node [style=none] (11) at (-0.5, 0.5) {$\Phi$};
		\node [style=none] (12) at (-0.5, 1) {};
		\node [style=none] (13) at (-0.5, -0.25) {};
		\node [style=none] (14) at (0.5, 1) {};
		\node [style=none] (15) at (0, 1.5) {};
		\node [style=none] (16) at (-0.5, 0) {};
		\node [style=none] (17) at (0.5, -0.5) {};
	\end{pgfonlayer}
	\begin{pgfonlayer}{edgelayer}
		\draw [style=empty] (2.center)
			 to (0.center)
			 to (1.center)
			 to (3.center)
			 to cycle;
		\draw [style=empty] (9.center)
			 to (7.center)
			 to (8.center)
			 to (10.center)
			 to cycle;
		\draw [style=arr] (5.center) to (13.center);
		\draw [style=arr, bend left=45] (12.center) to (15.center);
		\draw [style=arr] (14.center) to (6.center);
		\draw [bend left=45, looseness=0.75] (15.center) to (14.center);
		\draw (13.center) to (16.center);
		\draw (6.center) to (17.center);
	\end{pgfonlayer}
\end{tikzpicture} = 1
\end{array}
\right) \Rightarrow \left(
\begin{array}{c}
\exists \rho : I \to A \text{ causal} . \\
\begin{tikzpicture}
	\begin{pgfonlayer}{nodelayer}
		\node [style=none] (0) at (-1, 0.5) {};
		\node [style=none] (1) at (1, 0.5) {};
		\node [style=none] (2) at (-1, -0.5) {};
		\node [style=none] (3) at (1, -0.5) {};
		\node [style=none] (4) at (0, 0) {$w$};
		\node [style=none] (5) at (-0.5, 0.5) {};
		\node [style=none] (6) at (0.5, 0.75) {};
		\node [style=none] (13) at (-0.5, 0.75) {};
		\node [style=none] (14) at (0.5, 1) {};
		\node [style=none] (16) at (-0.5, 1) {};
		\node [style=none] (17) at (0.5, 0.5) {};
	\end{pgfonlayer}
	\begin{pgfonlayer}{edgelayer}
		\draw [style=empty] (2.center)
			 to (0.center)
			 to (1.center)
			 to (3.center)
			 to cycle;
		\draw [style=arr] (5.center) to (13.center);
		\draw [style=arr] (14.center) to (6.center);
		\draw (13.center) to (16.center);
		\draw (6.center) to (17.center);
	\end{pgfonlayer}
\end{tikzpicture} = \begin{tikzpicture}
	\begin{pgfonlayer}{nodelayer}
		\node [style=none] (0) at (-1, 0.5) {};
		\node [style=none] (1) at (0, 0.5) {};
		\node [style=none] (2) at (-1, -0.5) {};
		\node [style=none] (3) at (0, -0.5) {};
		\node [style=none] (4) at (-0.5, 0) {$\rho$};
		\node [style=none] (5) at (-0.5, 0.5) {};
		\node [style=none] (6) at (0.5, 0.5) {};
		\node [style=none] (13) at (-0.5, 0.75) {};
		\node [style=none] (14) at (0.5, 1) {};
		\node [style=none] (16) at (-0.5, 1) {};
		\node [style=maxmix] (17) at (0.5, 0) {};
	\end{pgfonlayer}
	\begin{pgfonlayer}{edgelayer}
		\draw [style=empty] (2.center)
			 to (0.center)
			 to (1.center)
			 to (3.center)
			 to cycle;
		\draw [style=arr] (5.center) to (13.center);
		\draw [style=arr] (14.center) to (6.center);
		\draw (13.center) to (16.center);
		\draw (6.center) to (17.center);
	\end{pgfonlayer}
\end{tikzpicture}
\end{array}
\right)
\end{equation} 
\end{enumerate}
\end{definition}

Note that in~\cite{Kissinger2019a}, \ref{prec:oneway} and \ref{prec:split} are rolled into a single axiom (C4), which is then proven equivalent to \ref{prec:oneway} and \ref{prec:split}.

\begin{example}\label{ex:classical}
  \MatRp is a precausal category, whose objects are natural numbers and whose morphisms $M: m \to n$ are $n \times m$ matrices. Then, $\otimes$ is given by tensor product (a.k.a. Kronecker product) of matrices and consequently the tensor unit is the natural number $1$. Hence states are column vectors, discarding maps $\discard_m: m \to 1$ are given by row vectors of all $1$'s, and the causality condition for states $\rho \fatsemi \discard = \id_1$ imposes the condition that the entries of $\rho$ sum to 1. The conditions \ref{prec:discard}, \ref{prec:dimension}, and \ref{prec:enough} are easily checked, whereas \ref{prec:oneway} and \ref{prec:split} follow from the product rule for conditional probability distributions (see \cite{Kissinger2019a}).
\end{example}

\begin{example}\label{ex:quantum}
  \CP is a precausal category, whose objects are algebras $\mathcal L(H)$ of linear operators from a finite-dimensional Hilbert space to itself, and whose morphisms are completely positive maps (CP-maps). $\otimes$ is given by tensor product and consequently the tensor unit is the 1D algebra $\mathcal L(\mathbb C) \cong \mathbb C$. States are CP-maps $\mathbb C \to \mathcal L(H)$, which correspond to positive semidefinite operators $\rho \in \mathcal L(H)$. Discarding is given by the trace, hence causal states are the trace-1 positive semidefinite operators, a.k.a. quantum (mixed) states.
  Again the conditions \ref{prec:discard}, \ref{prec:dimension}, and \ref{prec:enough} are easily checked, whereas \ref{prec:oneway} and \ref{prec:split} follow from the essential uniqueness of purification for CP-maps (see \cite{Kissinger2019a}).
\end{example}

Given a precausal category $\catc$, we can refine to a category $\caus[\catc]$ which equips each object with a nice set of states that should be considered ``normalised'' (or ``causal'') and restricts to the morphisms that preserve them. First, we'll define what it means for a set of states to be suitably nice.


\begin{definition}
A set $c \subseteq \catc(I, A)$ is \textit{closed} if $c = c^{**}$ and \textit{flat} if either there exist invertible scalars $\lambda, \mu$ such that $\lambda \cdot \maxmix_A \in c$ and $\mu \cdot \discard_A \in c^*$ or $A$ is a zero system.
\end{definition}

\begin{definition}
Given a precausal category $\catc$, the category $\caus[\catc]$ has as objects pairs $\mathbf{A} = (A, c_\mathbf{A} \subseteq \catc(I, A))$ where $c_\mathbf{A}$ is closed and flat. A morphism $f : \mathbf{A} \to \mathbf{B}$ is a morphism $f : A \to B$ in $\catc$ such that $\forall \rho \in c_\mathbf{A} . \rho \fatsemi f \in c_\mathbf{B}$.
\end{definition}

$\caus[\catc]$ inherits a monoidal structure and monoidal closure from $\catc$, from which one can show that it is $*$-autonomous \cite{Kissinger2019a}. This is in fact a full subcategory of a particular double gluing construction (specifically the tight orthogonality subcategory of the double glueing construction using the $\{1\}$-focussed orthogonality \cite{Hyland2003}). Aside from the flatness restriction, this is therefore a relatively well-known means of constructing models of linear logic.

\begin{tabular}{lrl}
First-order states & $\mathbf{A^1} :=$ & $\left(A, \left\{\discard_A\right\}^*\right)$ \\
Dual & $\mathbf{A^*} :=$ & $\left(A^*, c_\mathbf{A}^*\right)$ \\
Tensor product & $\mathbf{A \otimes B} :=$ & $\left(A \otimes B, \left\{\rho_A \otimes \rho_B \middle| \rho_A \in c_\mathbf{A}, \rho_B \in c_\mathbf{B} \right\}^{**}\right)$ \\
Par & $\mathbf{A \parr B} :=$ & $\mathbf{\left(A^* \otimes B^*\right)^*}$ \\
Internal hom & $\mathbf{A \multimap B} := $ & $ \mathbf{A^* \parr B}$ \\
Monoidal unit & $\mathbf{I} :=$ & $\left(I, \left\{1\right\}\right)$
\end{tabular}

The intuition between the two monoidal products is that $\mathbf{A \otimes B}$ is the closure of the space of local processes and hence we can compose the $\mathbf{A}$ and $\mathbf{B}$ components however we choose, whereas $\mathbf{A \parr B}$ is the space of bipartite processes that are normalised in local contexts so we generally cannot compose the components, just act locally on each side. Both share the same unit $\mathbf{I} = \mathbf{I^*}$ and there is a canonical inclusion $\mathbf{A \otimes B} \hookrightarrow \mathbf{A \parr B}$ making $\caus[\catc]$ into an isomix category.

From these operators, we can build objects capturing the set of processes compatible with some common causal structures. For example, $\mathbf{(A^1 \multimap B^1) \parr (C^1 \multimap D^1)}$ includes all causal bipartite first-order processes, $\mathbf{(A^1 \multimap B^1) \otimes (C^1 \multimap D^1)}$ is the subset of those that are non-signalling \cite[Theorem 6.2]{Kissinger2019a}, and $\mathbf{(A^1 \multimap B^1) \multimap (C^1 \multimap D^1)}$ includes all $2$-combs that map causal channels $\mathbf{A^1 \multimap B^1}$ to causal channels $\mathbf{C^1 \multimap D^1}$.


\section{Additive Precausal Categories}\label{sec:apc}

When investigating higher order causal theories, it is useful to strengthen the definition of a precausal category in a handful of ways. For the remainder of this paper, we will adapt the definition of $\caus[\catc]$ to be built from an additive precausal category $\catc$.

\begin{definition}\label{def:apc}
  Let $\catc$ be a compact closed category with products (and hence biproducts and additive enrichment~\cite{Houston2008}). $\catc$ is an \textit{additive precausal category} if:

  \begin{enumerate}[APC1.]
    \item[{\rm \namedlabel{apc:discard}{APC1}.}] $\catc$ has a discarding process $\discard_A \in \catc(A, I)$ for every system $A$, compatible with the monoidal and biproduct structures as below;
\begin{align}
\discard_{A \otimes B} &= \discard_A \otimes \discard_B \\
\discard_I &= \id_I \\
\discard_{A \oplus B} &= \langle \discard_A, \discard_B \rangle
\end{align}
\item[{\rm \namedlabel{apc:dimension}{APC2}.}] The dimension $d_A := \maxmix_A \fatsemi \discard_A$ is invertible for all non-zero $A$;
\item[{\rm \namedlabel{apc:basis}{APC3}.}] Each object $A \in \ob(\catc)$ has a finite causal basis: some $\{\rho_i\}_i \subseteq \left\{ \discard_A \right\}^*$ such that $\forall B . \forall f, g \in \catc(A,B) . (\forall i .\ \rho_i \fatsemi f = \rho_i \fatsemi g) \Rightarrow f = g$.
\item[{\rm \namedlabel{apc:scalars}{APC4}.}] Addition of scalars is \textit{cancellative} ($\forall x,y,z.\ x + z = y + z \Rightarrow x = y$), \textit{totally pre-ordered} ($\forall x,y .\exists z .\ x+z=y \vee x=y+z$), and all non-zero scalars have a multiplicative inverse.
\item[{\rm \namedlabel{apc:effects}{APC5}.}] All effects have a complement with respect to discarding: for any $\pi \in \catc(A,I)$, there exists some $\pi' \in \catc(A,I)$ and scalar $\lambda$ such that $\pi + \pi' = \lambda \cdot \discard_A$.
\end{enumerate}
\end{definition}

The first 3 axioms relate closely to the corresponding ones in Definition~\ref{def:precausal}, whereas the last two are quite different in flavour, and are in some sense more elementary, as their proofs don't rely on any particularly deep facts about our main classical and quantum examples (see Examples~\ref{ex:classical-apc} and \ref{ex:quantum-apc} below).



The axioms \ref{apc:discard} and \ref{apc:dimension} above are essentially identical to the corresponding axioms in Definition \ref{def:precausal} of precausal categories, with the additional requirement that discarding be compatible with biproducts as well as tensor products.

\ref{apc:basis} is a strengthening of condition \ref{prec:enough}. Rather than requiring us to check a pair of processes agree on \textit{all} causal states to be equal, we only require agreement on some fixed finite set of states. In other other words, each system has a set of states that behaves like a basis spanning all of the others, for the purposes of distinguishing maps. In the quantum foundations literature, this is sometimes called a \textit{fiducial set of states}.


\ref{apc:scalars} says that the semiring of scalars $\mathcal C(I,I)$ behaves somewhat like the set of non-negative real numbers $\mathbb R^+$. While the scalars need not be a field (indeed our main example $\mathbb R^+$ is not), any field satisfies this axiom as well. We do however exclude categories of non-determinstic processes such as $\mathrm{Rel}$, since the scalars are boolean values with addition given by the (non-cancellative) operation of disjunction.



Note that, with the help of bases (\ref{apc:basis}), we can promote additive cancellativity of scalars to additive cancellativity for all processes.

\begin{proposition}\label{prop:apc-cancellative}
  In an additive precausal category:
  \begin{equation}
    \tag{APC5a}
    \hfill
    \forall f, g, h \in \catc(A,B).\ 
    f + h = g + h \Rightarrow f = g
    \hfill
  \end{equation}
\end{proposition}

This condition allows us to define the free \textit{subtractive closure} $\sub(\catc)$ which extends $\catc$ with all negatives and prove that there exists a faithful embedding $\fsub{-} : \catc \to \sub(\catc)$.
More details on the explicit construction of $\sub(\catc)$, its properties, and the embedding functor are given in Appendix \ref{sec:subtraction}.

The utility of freely introducing negatives is summarised in the following proposition.

\begin{proposition}\label{prop:subc-field}
  For a precausal category $\catc$, the scalars $K := \sub(\catc)(I,I)$ are a field, and hence $\sub(\catc)$ is enriched over $K$-vector spaces.
\end{proposition}

In particular, we can therefore treat processes in $\catc(A,B)$ as being embedded in an ambient vector space $\sub(\catc)(A,B)$. So, while we might not be able to make all linear algebraic constructions directly in $\catc$, we can do so in $\sub(\catc)$. An important example of this is the ability to extend any independent set of states to a basis of states in $\catc$ as well as its corresponding dual basis of effects in $\sub(\catc)$.

\begin{lemma}\label{lemma:extend_basis}
Given any set of morphisms in $\catc(A,B)$ that are linearly independent in $\sub(\catc)$, they can be extended to a basis in $\catc$ with a dual basis in $\sub(\catc)$.
\end{lemma}

While the dual basis in $\sub(\catc)$ may not be physically meaningful (e.g. in the classical case it will contain vectors with negative probabilities), it will be a useful mathematical tool for working with the morphisms in $\catc$.

\ref{apc:effects} allows us to interpret effects (up to some renormalisation) as testing some predicate. To see how this works, first assume for simplicity that $\lambda = \id_I$. For a type $A$, we can think of $\pi: A \to I$ as some predicate over $A$, and $\pi'$ as its negation. For some causal state $\rho$, we can think of the composition $p_1 := \rho \fatsemi \pi$ as the probability that $\pi$ is true for $\rho$ and $p_2 := \rho \fatsemi \pi'$ as the probability that $\pi$ is false. The fact that $\pi + \pi' = \discard$ lets us conclude that those probabilities sum to 1:
\[
  \hfill
  p_1 + p_2 = \rho \fatsemi \pi + \rho \fatsemi \pi' = \rho \fatsemi (\pi + \pi') = \rho \fatsemi \discard = \id_I
  \hfill
\]
If $\lambda \neq \id_I$, the previous reasoning holds after re-normalising, i.e. replacing $\pi$ and $\pi'$ with $\lambda^{-1} \cdot \pi$ and $\lambda^{-1} \cdot \pi'$.

Thanks to compact closure, we can promote \ref{apc:effects} to a property about all processes. Namely, any process $f : A \to B$ has a complement $f'$ where, up to re-normalisation, $f + f'$ gives the uniform noise process.

\begin{proposition}\label{prop:bintests}
  For any $f: A \to B$ in an additive precausal category, there exists $f' : A \to B$ and a scalar $\lambda$ such that:
  \begin{equation}
    \tag{APC5a}
    \hfill
    f + f' = \lambda \cdot \discard_A \fatsemi \maxmix_B
    \hfill
  \end{equation}
\end{proposition}

\begin{example}\label{ex:classical-apc}
  \MatRp defined as in Example~\ref{ex:classical} is also an additive precausal category, where $\oplus$ is given by the direct sum of matrices. The standard basis of unit vectors gives a basis for \ref{apc:basis}, the semiring of scalars $\MatRp(I,I) \cong \mathbb R^+$ satisfies \ref{apc:scalars}, and for \ref{apc:effects}, we just need to choose a suitably large $\lambda$ such that $\pi' := \lambda \cdot \discard_A - \, \pi$ contains all positive numbers.
\end{example}

\begin{example}\label{ex:affine-apc}
  In addition to $\mathbb R^+$, we can construct an additive precausal category $\textrm{Mat}[K]$ for any field of characteristic 0. In particular, \MatRaff is an additive precausal category that is identical to \MatRp but without any positivity constraint, describing affine or ``quasi-probabilistic'' maps where negative probabilities are permitted. In this case, the subtractive closure gives an equivalent category to \MatRaff itself.
\end{example}

\begin{example}\label{ex:quantum-apc}
  The quantum example is very nearly the category \CP, as defined in Example~\ref{ex:quantum}, but \CP doesn't have biproducts. If we freely add biproducts, we obtain a category \CPs whose objects are all finite-dimensional C*-algebras (or equivalently, algebras of the form $\mathcal L(H_1) \oplus \ldots \oplus \mathcal L(H_k)$) and completely positive maps. Discarding is again given by the trace operator, so \ref{apc:discard} and \ref{apc:dimension} are straightforward to verify. For \ref{apc:basis}, we can fix a (non-orthogonal) basis of states for each type. As in the classical case, the scalars are $\mathbb R^+$, so \ref{apc:scalars} is immediate and since $\discard_A$ is an interior point in the cone of positive effects, $\pi'$ can defined as $\lambda \cdot \discard_A - \, \pi$ for suitably large $\lambda$.
\end{example}

%

Given an additive precausal category, we can apply the $\caus$ construction in the same way as before. However, it may now be easier to devise interesting closed sets or interpret the impact of the closure operator since it just corresponds to taking affine combinations of states. For this to make sense, we should say precisely what we mean to take affine combinations of states in $\catc$.

\begin{definition}\label{def:affine-closure}
  For a set of states $c \subseteq \catc(I, A)$, we define sets $\aff(c) \subseteq \sub(\catc)(I, A)$ and $\affp(c) \subseteq \catc(I,A)$ as follows, for $K := \sub(\catc)(I,I)$:
  \begin{align*}
    \aff(c) & := \left\{ \rho \, \middle| \, \exists \{\rho_i\}_i \subseteq \catc(I, A), \{\lambda_i\}_i \subseteq K . \ \sum_i \lambda_i = \id_I, \rho = \sum_i \lambda_i \cdot \fsub{\rho_i} \right\} \\
    \affp(c) & := \left\{ \rho \, \middle| \, \exists \rho' \in \aff(c).\ \rho' = \fsub{\rho} \right\}
  \end{align*}
\end{definition}

If we identify the set $\catc(I,A)$ with its image under $\fsub{-}$, we can think if $\affp(c)$ as the intersection of the affine closure of $c$ with the set $\catc(I,A) \subseteq \sub(\catc)(I,A)$ of ``positive'' states embedded in the subtractive closure. In the classical and quantum case, $\affp(-)$ arises from taking all of the affine combinations of elements of $c$, then intersecting the resulting set with the positive cone of (unnormalised) probability distributions or quantum states, respectively.


\begin{theorem}\label{thrm:affine}
Given any flat set $c \subseteq \catc(I,A)$ for a non-zero $A$, $c^{**} = \affp(c)$.
\end{theorem}

This characterises the tensor space $c_{\mathbf{A \otimes B}} = \left\{ \rho_A \otimes \rho_B \middle| \rho_A \in c_\mathbf{A}, \rho_B \in c_\mathbf{B} \right\}^{**}$ in $\caus[\catc]$ as the affine closure of the separable states, from which we can prove the following property.

\begin{theorem}\label{thrm:no_interaction_with_trivial}
If $\mathbf{A} = (A, c_\mathbf{A})$ with $c_\mathbf{A} = \{\mu \cdot \maxmix_A\}$ for any non-zero $A$, then every $h \in c_{\mathbf{A \otimes B}}$ is a product morphism of the form $\mu \cdot \maxmix_A \otimes g$ for some $g \in c_\mathbf{B}$.
\end{theorem}

This captures what \cite{Wilson2021a} refers to as the principle of ``no interaction with trivial degrees of freedom''. In particular, it recovers the precausal category axiom \ref{prec:split} by showing that every state of $\mathbf{(A^1 \multimap B^1)^*} = \mathbf{({A^1}^* \parr B^1)^*} = \mathbf{A^1 \otimes {B^1}^*}$ decomposes into a product of a state of $\mathbf{A^1}$ (i.e. a causal state) and $\maxmix_{B^*}$.


It should be noted that we have completely dropped condition \ref{prec:oneway} on our underlying category of basic processes. In Section \ref{sec:oneway} we will recover a slightly weaker version of this condition by the equivalence of one-way signalling and the affine closure of semi-localisability (Theorem \ref{thrm:one_way_self_dual_and_affine_semi_local}). Fortunately, we can still reuse the same proof to show that $\caus[\catc]$ is $*$-autonomous for additive precausal $\catc$ since it doesn't rely on \ref{prec:oneway}.


\section{Additive Types}\label{sec:additive}

We may also lean on the relation to the double glueing construction to add type constructors corresponding to the additives of linear logic. In categorical models of linear logic, additive conjunction of types is captured by cartesian product and additive disjunction by coproduct, satisfying a De Morgan duality with one another \cite{Mellies}. In terms of resources, these represent a classical choice: an instance of $A \times B$ is a single resource unit that we can choose to be used either as an instance of $A$ or an instance of $B$, whereas an instance of $A + B$ is a single resource unit which is fixed on creation as either a unit of $A$ or a unit of $B$.

The concept of classical choice is often built into operational theories in the form of probability distributions, classical outcomes of tests, and conditional tests. This is typically deemed essential for any experimentalist who is bound to classical data to interact with and make inference from an experiment. On the other hand, the $\caus$ construction concerns the underlying systems present in the physical theory without consideration of any classical agent. We can recover finite-outcome random variables by incorporating binary classical choice through new additive type constructors, i.e. finding constructions for products and coproducts in $\caus[\catc]$.

\begin{definition}
Given types $\mathbf{A} = (A, c_\mathbf{A}), \mathbf{B} = (B, c_{\mathbf{B}})$ in $\caus[\catc]$, we define $\mathbf{A \times B} := (A \oplus B, c_{\mathbf{A \times B}})$ and $\mathbf{A \oplus B} := (A \oplus B, c_{\mathbf{A \oplus B}})$ where:

\begin{equation}\label{eq:def_product}
\begin{split}
c_{\mathbf{A \times B}} &= \left( \{p_A \fatsemi \pi_A | \pi_A \in c_\mathbf{A}^* \subseteq \catc(A,I) \} \cup \{p_B \fatsemi \pi_B | \pi_B \in c_\mathbf{B}^* \subseteq \catc(B,I) \} \right)^* \\
&= \left\{ (\rho_A, \rho_B) | \rho_A \in c_\mathbf{A}, \rho_B \in c_\mathbf{B} \right\}
\end{split}
\end{equation}
\begin{equation}\label{eq:def_coproduct}
\begin{split}
c_{\mathbf{A \oplus B}} &= \left( \{ \rho_A \fatsemi \iota_A | \rho_A \in c_\mathbf{A} \} \cup \{ \rho_B \fatsemi \iota_B | \rho_B \in c_\mathbf{B} \} \right)^{**} \\
&= \left\{ \langle \pi_A, \pi_B \rangle | \pi_A \in c_\mathbf{A}^* \subseteq \catc(A,I), \pi_B \in c_\mathbf{B}^* \subseteq \catc(B,I) \right\}^*
\end{split}
\end{equation}
\end{definition}

\begin{lemma}\label{lemma:additive_defs_equivalent}
The alternative definitions of $c_{\mathbf{A \times B}}$ and $c_{\mathbf{A \oplus B}}$ are equivalent; that is, Equations \ref{eq:def_product} and \ref{eq:def_coproduct} hold.
\end{lemma}

\begin{corollary}\label{corollary:additive_demorgan}
The operators $\times$ and $\oplus$ are De Morgan duals under $(-)^*$.
\end{corollary}

\begin{proposition}\label{prop:product}
$\mathbf{A \times B}$ is a categorical product in $\caus[\catc]$.
\end{proposition}

\begin{proposition}\label{prop:coproduct}
$\mathbf{A \oplus B}$ is a categorical coproduct in $\caus[\catc]$.
\end{proposition}

The zero object $0$ has a unique state $0_{I,0} \in \catc(I,0)$ by terminality and a unique effect $0_{0,I} \in \catc(0,I)$ by initiality. There are only two candidates for causal sets: the empty set $\emptyset$ and the singleton $\{0_{I,0}\}$. These have roles in our causal category as the additive units.

\begin{proposition}\label{prop:initial_terminal}
The initial object in $\caus[\catc]$ is $\mathbf{0} := (0, \emptyset)$ and the terminal object is $\mathbf{1} := (0, \{0_{I,0}\})$. Furthermore, they are duals of each other and are units for $\mathbf{\oplus}$ and $\mathbf{\times}$ respectively.
\end{proposition}

Thinking of first-order types as describing systems with no input (i.e. no choice in how to consume them), both the product and coproduct have interesting interactions with first-order types because of where the classical choice happens. For coproducts, the choice is already fixed in the creation of a state so we expect it to preserve the first-order property. However, products introduce freedom of choice in effects allowing us to view the projectors as inputs to the system dictating whether it should prepare the left or the right state.



\begin{proposition}\label{prop:coproduct_first_order}
If $\mathbf{A}$ and $\mathbf{B}$ are both first-order types, then so is $\mathbf{A \oplus B}$.
\end{proposition}

\begin{proposition}\label{prop:product_not_first_order}
If $A$ and $B$ are non-zero, then $\mathbf{A \times B}$ is never a first-order type.
\end{proposition}

We have both $*$-autonomy and all finite products and coproducts, which is all the evidence required to show that $\caus[\catc]$ is a model of linear logic with additives. For both the multiplicative and additive structures, we found that the $\caus$ construction took a degenerative categorical structure and generated non-degenerate structures from them. For example, the compact closure of $\catc$ means that it is a $*$-autonomous category in which the two monoidal products are the same, but $\caus[\catc]$ is $*$-autonomous with distinct $\mathbf{\otimes}$ and $\mathbf{\parr}$. Similarly, the construction for the additives takes biproducts/the zero object and yields distinct products and coproducts/initial and terminal objects. The degenerate exception here is that $\mathbf{I}$ is still the unit for both $\mathbf{\otimes}$ and $\mathbf{\parr}$.

\section{One-Way Signalling Types}\label{sec:oneway}

When examining multi-partite systems, causal structures can be investigated from the perspective of information signalling (which parties can observe changes to another party's inputs) or decompositions (does the channel admit a decomposition into local channels compatible with some configuration of time- and space-like separations between the parties). In the bipartite case, we can compare these perspectives with the examples of one-way signalling and semi-localisable channels.

The one-way signalling causal structure refers to a bipartite system (say, between Alice and Bob) in which causal influence may only exist in one direction. There is a standard definition given in causality literature for one-way signalling for quantum channels.

\begin{definition}\label{def:OneWay}
A bipartite process is one-way signalling (Alice to Bob) if discarding Bob's output admits a decomposition into local processes.

\begin{equation*}
\hfill
\begin{tikzpicture}
	\begin{pgfonlayer}{nodelayer}
		\node [style=none] (0) at (0, -0.5) {};
		\node [style=none] (1) at (0, 0.5) {};
		\node [style=none] (2) at (2, 0.5) {};
		\node [style=none] (3) at (2, -0.5) {};
		\node [style=none] (4) at (0.5, -0.5) {};
		\node [style=none] (5) at (0.5, -1) {};
		\node [style=none] (6) at (0.5, 0.5) {};
		\node [style=none] (7) at (0.5, 1) {};
		\node [style=none] (8) at (1.5, 0.5) {};
		\node [style=discard] (9) at (1.5, 1) {};
		\node [style=none] (10) at (1.5, -0.5) {};
		\node [style=none] (11) at (1.5, -1) {};
		\node [style=none] (12) at (1, 0) {$\psi$};
	\end{pgfonlayer}
	\begin{pgfonlayer}{edgelayer}
		\draw [style=empty] (3.center)
			 to (0.center)
			 to (1.center)
			 to (2.center)
			 to cycle;
		\draw (5.center) to (4.center);
		\draw (6.center) to (7.center);
		\draw (11.center) to (10.center);
		\draw (8.center) to (9);
	\end{pgfonlayer}
\end{tikzpicture} = \begin{tikzpicture}
	\begin{pgfonlayer}{nodelayer}
		\node [style=none] (0) at (0, -0.5) {};
		\node [style=none] (1) at (0, 0.5) {};
		\node [style=none] (2) at (1, 0.5) {};
		\node [style=none] (3) at (1, -0.5) {};
		\node [style=none] (4) at (0.5, 0) {$\psi_A$};
		\node [style=none] (5) at (0.5, -0.5) {};
		\node [style=none] (6) at (0.5, -1) {};
		\node [style=none] (7) at (0.5, 0.5) {};
		\node [style=none] (8) at (0.5, 1) {};
		\node [style=none] (9) at (1.5, -1) {};
		\node [style=discard] (10) at (1.5, 0) {};
	\end{pgfonlayer}
	\begin{pgfonlayer}{edgelayer}
		\draw [style=empty] (3.center)
			 to (0.center)
			 to (1.center)
			 to (2.center)
			 to cycle;
		\draw (6.center) to (5.center);
		\draw (7.center) to (8.center);
		\draw (9.center) to (10);
	\end{pgfonlayer}
\end{tikzpicture}
\hfill
\end{equation*}
\end{definition}

Semi-localisability is an alternative to one-way signalling for expressing compatibility with a time-like separation of parties, giving a constructive example of how the combined channel can be decomposed into local operations.

\begin{definition}
A bipartite channel between Alice and Bob is semi-localisable (with Alice before Bob) if it decomposes into local processes with a channel from Alice to Bob.

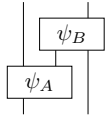
\begin{center}
\begin{tikzpicture}
	\begin{pgfonlayer}{nodelayer}
		\node [style=none] (0) at (0, 0) {};
		\node [style=none] (1) at (0, 1) {};
		\node [style=none] (2) at (2, 1) {};
		\node [style=none] (3) at (2, 0) {};
		\node [style=none] (4) at (1, 1.5) {};
		\node [style=none] (5) at (1, 2.5) {};
		\node [style=none] (6) at (3, 2.5) {};
		\node [style=none] (7) at (3, 1.5) {};
		\node [style=none] (8) at (1, 0.5) {$\psi_A$};
		\node [style=none] (9) at (2, 2) {$\psi_B$};
		\node [style=none] (10) at (0.5, -0.5) {};
		\node [style=none] (11) at (0.5, 0) {};
		\node [style=none] (12) at (0.5, 1) {};
		\node [style=none] (13) at (0.5, 3) {};
		\node [style=none] (14) at (2.5, -0.5) {};
		\node [style=none] (15) at (2.5, 1.5) {};
		\node [style=none] (16) at (2.5, 2.5) {};
		\node [style=none] (17) at (2.5, 3) {};
		\node [style=none] (18) at (1.5, 1) {};
		\node [style=none] (19) at (1.5, 1.5) {};
	\end{pgfonlayer}
	\begin{pgfonlayer}{edgelayer}
		\draw [style=empty] (0.center)
			 to (1.center)
			 to (2.center)
			 to (3.center)
			 to cycle;
		\draw [style=empty] (4.center)
			 to (5.center)
			 to (6.center)
			 to (7.center)
			 to cycle;
		\draw (10.center) to (11.center);
		\draw (12.center) to (13.center);
		\draw (18.center) to (19.center);
		\draw (14.center) to (15.center);
		\draw (16.center) to (17.center);
	\end{pgfonlayer}
\end{tikzpicture}
\end{center}
\end{definition}

Both of these properties state that the channel is compatible with the setting where Bob is in Alice's future light cone. It is important to note that both are judgements of compatibility with a causal structure rather than inference of any necessary causal relationship between Alice and Bob - completely local processes trivially satisfy these properties but obviously have no causal influence between the parties.

These definitions are specific to first-order channels, but we can consider situations where Alice and Bob have higher-order systems $\mathbf{A}, \mathbf{B} \in \ob(\caus[\catc])$ representing some more complex interaction with their environments. We will describe these properties as constructions on the causal sets.

The essence of one-way signalling is that Alice's local effect is independent of any input given to Bob by his context:
\begin{equation}
  \hfill
  c_\mathbf{A} < c_\mathbf{B} := \left\{ h \in c_{\mathbf{A \parr B}} \middle| \exists m \in c_\mathbf{A} . \forall \pi \in c_{\mathbf{B}^*} . \begin{tikzpicture}
	\begin{pgfonlayer}{nodelayer}
		\node [style=none] (0) at (0, 0) {};
		\node [style=none] (1) at (2, 0) {};
		\node [style=none] (2) at (2, -1) {};
		\node [style=none] (3) at (0, -1) {};
		\node [style=none] (4) at (1, -0.5) {$h$};
		\node [style=none] (5) at (3, 0) {};
		\node [style=none] (6) at (4, 0) {};
		\node [style=none] (7) at (4, -1) {};
		\node [style=none] (8) at (3, -1) {};
		\node [style=none] (9) at (3.5, -0.5) {$\pi$};
		\node [style=none, label={[label distance=-3pt]45:$B^*$}] (10) at (3.5, 0) {};
		\node [style=none, label={[label distance=-3pt]45:$B$}] (11) at (1.5, 0) {};
		\node [style=none, label={[label distance=-3pt]45:$A$}] (12) at (0.5, 0) {};
		\node [style=none] (13) at (2.5, 1) {};
		\node [style=none] (14) at (0.5, 1) {};
		\node [style=none] (15) at (0.5, 0.5) {};
	\end{pgfonlayer}
	\begin{pgfonlayer}{edgelayer}
		\draw [style=empty] (0.center)
			 to (1.center)
			 to (2.center)
			 to (3.center)
			 to cycle;
		\draw [style=empty] (5.center)
			 to (6.center)
			 to (7.center)
			 to (8.center)
			 to cycle;
		\draw [style=arr] (12.center) to (15.center);
		\draw [style=arr, bend left=45] (11.center) to (13.center);
		\draw [bend left=45] (13.center) to (10.center);
		\draw (15.center) to (14.center);
	\end{pgfonlayer}
\end{tikzpicture} = \begin{tikzpicture}
	\begin{pgfonlayer}{nodelayer}
		\node [style=none] (0) at (0, 0) {};
		\node [style=none] (1) at (1, 0) {};
		\node [style=none] (2) at (1, -1) {};
		\node [style=none] (3) at (0, -1) {};
		\node [style=none] (4) at (0.5, -0.5) {$m$};
		\node [style=none, label={[label distance=-3pt]45:$A$}] (12) at (0.5, 0) {};
		\node [style=none] (14) at (0.5, 1) {};
		\node [style=none] (15) at (0.5, 0.5) {};
	\end{pgfonlayer}
	\begin{pgfonlayer}{edgelayer}
		\draw [style=empty] (0.center)
			 to (1.center)
			 to (2.center)
			 to (3.center)
			 to cycle;
		\draw [style=arr] (12.center) to (15.center);
		\draw (15.center) to (14.center);
	\end{pgfonlayer}
\end{tikzpicture} \right\}
  \hfill
\end{equation}
When this property holds, we refer to $m$ as the \textit{(left-)residual} of $h$ (with \textit{right-residual} for the corresponding component in the symmetrically-defined $c_\mathbf{A} > c_\mathbf{B}$). The residual represents the constant local effect Alice observes regardless of inputs provided to Bob. In the case where $\mathbf{A}$ and $\mathbf{B}$ both describe first-order channels, this exactly reduces to Definition \ref{def:OneWay}, since the only causal contexts for a causal channel is to provide an arbitrary state at the input and discard the output.

For semi-localisability, it is not enough to consider a factorisation by any system, but specifically by a first-order system:
\begin{equation}
  \hfill
  c_\mathbf{A} \triangleleft c_\mathbf{B} := \left\{h \in c_{\mathbf{A \parr B}} \middle| \begin{array}{l} \exists \mathbf{Z} = \left( Z, \left\{ \discard_Z \right\}^* \right), m \in c_{\mathbf{A \parr Z}}, n \in c_{\mathbf{Z^* \parr B}} . \\ \begin{tikzpicture}
	\begin{pgfonlayer}{nodelayer}
		\node [style=none] (0) at (0, 0) {};
		\node [style=none] (1) at (2, 0) {};
		\node [style=none] (2) at (2, -1) {};
		\node [style=none] (3) at (0, -1) {};
		\node [style=none] (4) at (1, -0.5) {$h$};
		\node [style=none, label={[label distance=-3pt]45:$B$}] (11) at (1.5, 0) {};
		\node [style=none, label={[label distance=-3pt]45:$A$}] (12) at (0.5, 0) {};
		\node [style=none] (13) at (1.5, 0.5) {};
		\node [style=none] (14) at (0.5, 1) {};
		\node [style=none] (15) at (0.5, 0.5) {};
		\node [style=none] (16) at (1.5, 1) {};
	\end{pgfonlayer}
	\begin{pgfonlayer}{edgelayer}
		\draw [style=empty] (0.center)
			 to (1.center)
			 to (2.center)
			 to (3.center)
			 to cycle;
		\draw [style=arr] (12.center) to (15.center);
		\draw (15.center) to (14.center);
		\draw [style=arr] (11.center) to (13.center);
		\draw (13.center) to (16.center);
	\end{pgfonlayer}
\end{tikzpicture} = \begin{tikzpicture}
	\begin{pgfonlayer}{nodelayer}
		\node [style=none] (0) at (0, 0) {};
		\node [style=none] (1) at (2, 0) {};
		\node [style=none] (2) at (2, -1) {};
		\node [style=none] (3) at (0, -1) {};
		\node [style=none] (4) at (1, -0.5) {$m$};
		\node [style=none] (5) at (3, 0) {};
		\node [style=none] (6) at (5, 0) {};
		\node [style=none] (7) at (5, -1) {};
		\node [style=none] (8) at (3, -1) {};
		\node [style=none] (9) at (4, -0.5) {$n$};
		\node [style=none, label={[label distance=-3pt]45:$Z^*$}] (10) at (3.5, 0) {};
		\node [style=none, label={[label distance=-3pt]45:$Z$}] (11) at (1.5, 0) {};
		\node [style=none, label={[label distance=-3pt]45:$A$}] (12) at (0.5, 0) {};
		\node [style=none] (13) at (2.5, 1) {};
		\node [style=none] (14) at (0.5, 1) {};
		\node [style=none] (15) at (0.5, 0.5) {};
		\node [style=none, label={[label distance=-3pt]45:$B$}] (16) at (4.5, 0) {};
		\node [style=none] (17) at (4.5, 0.5) {};
		\node [style=none] (18) at (4.5, 1) {};
	\end{pgfonlayer}
	\begin{pgfonlayer}{edgelayer}
		\draw [style=empty] (0.center)
			 to (1.center)
			 to (2.center)
			 to (3.center)
			 to cycle;
		\draw [style=empty] (5.center)
			 to (6.center)
			 to (7.center)
			 to (8.center)
			 to cycle;
		\draw [style=arr] (12.center) to (15.center);
		\draw [style=arr, bend left=45] (11.center) to (13.center);
		\draw [bend left=45] (13.center) to (10.center);
		\draw (15.center) to (14.center);
		\draw [style=arr] (16.center) to (17.center);
		\draw (17.center) to (18.center);
	\end{pgfonlayer}
\end{tikzpicture} \end{array} \right\}
  \hfill
\end{equation}
This is because using a higher-order system (e.g. the dual of a first-order system) may allow information to flow in the opposite direction. Again, we recover the original definition of semi-localisability if we fix $\mathbf{A}$ and $\mathbf{B}$ to types of first-order channels (up to our encoding of channels via the Choi-Jamio\l{}kowski isomorphism).

Intuitively, based on the interpretation of the sequence operator in a $\bv$-category as representing an ordered combination of systems, both of these would be good candidates for sequence operator in $\caus[\catc]$. It is also interesting to compare them to sequence operators from other $\bv$-categories. The following construction is inspired by that of the $\bv$-category of probabilistic coherence spaces \cite{Blute2012}:
\begin{equation}
  \hfill
  c_\mathbf{A} \obslash c_\mathbf{B} := \left\{ h \in c_{\mathbf{A \parr B}} \middle| \begin{array}{l} \exists \mathcal{I}, \{f_i\}_{i \in \mathcal{I}} \subseteq \sub(\catc)(I, A \otimes B), \{g_i\}_{i \in \mathcal{I}} \subseteq c_\mathbf{B}, f \in c_\mathbf{A} . \\ \fsub{h} \sim \sum_{i \in \mathcal{I}} f_i \otimes \fsub{g_i} \wedge \fsub{f} \sim \sum_{i \in \mathcal{I}} f_i \end{array} \right\}
  \hfill
\end{equation}
Note this definition refers to maps in the subtractive closure $\sub(\catc)$ (cf. Appendix~\ref{sec:subtraction}) to make use of the affine-linear structure.



A core result in the field of causal structures is the equivalence of one-way signalling and semi-localisability for first-order channels \cite{Chiribella2010}. In the setting provided by the $\caus$ construction, we find a higher-order generalisation of this, where all three of these definitions coincide up to affine closure, as well as a proof that they are self-dual properties.


\begin{theorem}\label{thrm:one_way_self_dual_and_affine_semi_local}
$c_\mathbf{A} < c_\mathbf{B} = \left(c_\mathbf{A}^* < c_\mathbf{B}^*\right)^* = \left(c_\mathbf{A} \triangleleft c_\mathbf{B}\right)^{**} = c_\mathbf{A} \obslash c_\mathbf{B}$
\end{theorem}

$c_\mathbf{A} < c_\mathbf{B}$ is closed, since it the dual of another set, and flat, since $c_{\mathbf{A \otimes B}} \subseteq c_\mathbf{A} < c_\mathbf{B}$ gives the uniform state and $c_\mathbf{A} < c_\mathbf{B} \subseteq c_{\mathbf{A \parr B}}$ gives discarding. We can therefore elevate it to a genuine object $\mathbf{A < B} := (A \otimes B, c_\mathbf{A} < c_\mathbf{B})$ in $\caus[\catc]$.

We can go further into the categorical structure induced by this type constructor and show that it adds another monoidal structure with a weak interchange with both $\mathbf{\otimes}$ and $\mathbf{\parr}$, turning $\caus[\catc]$ into a model of $\bv$-logic.



\begin{definition}
Given a symmetric, linearly distributive category $\mathcal{D}$, a \textit{weak interchange} is an additional monoidal structure $(\mathcal{D}, \obslash, I_\obslash)$ with natural transformations
\begin{equation*}
  \begin{array}{rccrc}
    w_\otimes &: (R \obslash U) \otimes (T \obslash V) \to (R \otimes T) \obslash (U \otimes V) &
    \qquad\qquad &
    w_{I_\otimes} &: I_\otimes \to I_\otimes \obslash I_\otimes \\
    w_\parr &: (C \parr E) \obslash (D \parr F) \to (C \obslash D) \parr (E \obslash F) &
    \qquad\qquad &
    w_{I_\parr} &: I_\parr \obslash I_\parr \to I_\parr
  \end{array}
\end{equation*}
which are compatible with the structure isomorphisms of $(\mathcal{D}, \otimes, I_\otimes)$ and $(\mathcal{D}, \parr, I_\parr)$ and the distributive structure (we refer the reader to Blute, Panangaden, and Slavnov \cite[Definition 4.1]{Blute2012} for the full set of conditions).

A \textit{$\bv$-category} is a symmetric linearly distributive category with a weak interchange and an isomorphism $m : I_\obslash \to I_\otimes$ such that $m$ is an isomix map and $m^{-1}$ is a counit for $w_{I_\otimes}$:
\begin{equation*}
\begin{split}
(\id_{I_\obslash} \otimes m) \fatsemi \rho_{I_\obslash}^\otimes = (m \otimes \id_{I_\obslash}) \fatsemi \lambda_{I_\obslash}^\otimes : I_\obslash \otimes I_\obslash \to I_\obslash\\
w_{I_\otimes} \fatsemi (m^{-1} \obslash \id_{I_\otimes}) \fatsemi \lambda_{I_\otimes}^\obslash = \id_{I_\otimes} = w_{I_\otimes} \fatsemi (\id_{I_\otimes} \obslash m^{-1}) \fatsemi \rho_{I_\otimes}^\obslash : I_\otimes \to I_\otimes
\end{split}
\end{equation*}
\end{definition}

\begin{theorem}\label{thrm:bv_category}
$\caus[\catc]$ is a $\bv$-category.
\end{theorem}

The equivalence of one-way signalling and (the affine closure of) semi-localisability show that any directed information signalling will always exhibit an equivalent construction where the information transfer is mediated by first-order types. Moreover, first-order types (or scaled versions thereof) are \textit{exactly} those that can carry information in one direction.

\begin{theorem}\label{thrm:first_order_characterisation}
$\mathbf{A^* \parr A} = \mathbf{A^* < A} \Leftrightarrow \exists \mu \neq 0 . c_\mathbf{A}^* = \{\mu \cdot \discard_A\}$.
\end{theorem}

This result presents a characterisation of first-order system types that can be interpreted in any $\bv$-category, which may be a useful lemma in characterising which categories arise as $\caus[\catc]$ for some additive precausal $\catc$.

It is also interesting to compare this result to the causality principle: that an operational theory has no signalling from the future iff every system type has a unique deterministic effect \cite{Chiribella2010}, i.e. is a first-order type. One would think this would also give a characterisation of first-order types, but it contains a subtle exception for zero objects where the existence of an effect might fail. This is not particularly concerning for practical theories since the zero objects have little pragmatic importance, but does complicate the usage of this as a logical characterisation of first-order systems. The following shows this in the case of binary test outcomes encoded in $\mathbf{I \oplus I}$, but the same proof works if we replace it with any $\mathbf{B}$ with at least two distinct states (using a dual basis from Lemma \ref{lemma:extend_basis} to fill the role of the injectors and projectors).

\begin{proposition}\label{prop:causality_exception}
$\mathbf{(I \oplus I) \parr A} = \mathbf{(I \oplus I) < A} \Leftrightarrow \mathbf{A} = \mathbf{1} \vee \exists \mu \neq 0 . c_\mathbf{A}^* = \{\mu \cdot \discard_A\}$
\end{proposition}






\section{Non-Signalling Systems}\label{sec:nonsignalling}

The one-way signalling structure is not the only interesting causal structure on a bipartite system. For example, the non-signalling causal structure is a symmetric property about the statistical independence of the different parties.

\begin{definition}
A bipartite process is non-signalling if it satisfies the one-way signalling condition in both directions.
\end{definition}

This is a necessary condition for compatibility with the setting where Alice and Bob are space-like separated. However, this is not a sufficient condition since there exist non-signalling processes, such as a quantum implementation of a PR box, that cannot be factorised into local processes with a shared history \cite{Popescu1994}.

In~\cite{Kissinger2019a}, the tensor product of first-order channel types $\mathbf{(A^1 \multimap B^1) \otimes (C^1 \multimap D^1)}$ was shown to exactly contain those bipartite channels that are non-signalling \cite{Kissinger2019a}. This proof relied on some properties that do not generalise beyond first-order channels, such as the ability to apply \ref{prec:split} to decompose causal contexts for channels into a causal input state and discarding the output. Between Gutoski \cite{Gutoski2008} and Chiribella et al. \cite{Chiribella2013}, it was shown that the space of non-signalling channels in quantum theory can be characterised as the affine closure of product channels. By translating this proof into categorical terms, we can generalise the result to hold for arbitrary higher-order systems in $\caus[\catc]$ for any additive precausal $\catc$.

\begin{theorem}\label{thrm:tensor_nonsignalling}
$(c_\mathbf{A} < c_\mathbf{B}) \cap (c_\mathbf{A} > c_\mathbf{B}) = c_{\mathbf{A \otimes B}}$.
\end{theorem}


From this, we can construct a model of Barrett's Generalised Non-Signalling Theory \cite{Barrett2007} within $\caus[\MatRp]$. States are vectors of positive real numbers describing the probabilities of each outcome in some set of fiducial measurements (i.e. a finite set of measurements with finite outcomes that are sufficient to characterise the state, fixed for each system up front). The set of allowed states of a single system include any vector that gives a normalised probability distribution for each fiducial measurement. An $(n,k)$-system ($n$ fiducial measurements, each with $k$ outcomes) is characterised by the type $\mathbf{(I^{\oplus k})^{\times n}}$, since states are a vector of $n \cdot k$ positive reals and the dual space of effects is generated by $\left\{p_i \fatsemi \discard_{I^{\oplus k}}\right\}_i$ (choosing a fiducial measurement and marginalising over the outcomes). For multi-partite systems, the permissible states require both the normalisation condition and an extra non-signalling condition, matching the characterisation of $\mathbf{A \otimes B}$ as the space of non-signalling processes. In terms of transformations, the theory allows every matrix that preserves the normalisation condition on all input states, i.e. yields $1$ for every normalised state provided as input and normalised effect applied to the output. There is no equivalent realisation of the Generalised Local Theory (also from \cite{Barrett2007}) since the space of local multi-partite states is not affine-closed and therefore will not form a closed causal set.

\section{Conclusion}\label{sec:conclusion}

By extending the assumptions on the base category with additive structure, 
the $\caus$ construction yields a $\bv$-category with additives within which the characterisation of first-order types resembles the causality condition expressed as an equation of types. We also obtain general characterisations of $(-)^{**}$ as affine closure and the tensor product as the space of non-signalling bipartite processes. A number of the proofs are similar to those used to show that probabilistic coherence spaces form a $\bv$-category \cite{Blute2012}, 
and the characterisation of non-signalling processes as affine closure of local processes for first-order channels in quantum theory \cite{Gutoski2008}.

We recover almost all of the precausal conditions from an additive precausal category. Instead of the equivalence of one-way signalling and semi-localisability for first-order channels (Condition \ref{prec:oneway}), we obtained equivalence with the affine closure of semi-localisability for arbitrary higher-order systems (Theorem \ref{thrm:one_way_self_dual_and_affine_semi_local}). Since existing proofs of \ref{prec:oneway} for precausal categories rely on the essential uniqueness of purification, it may be possible to strengthen this result in future work by incorporating notions of purity and purification into this framework.

There are still other interesting process theories that cannot be considered as either the base category or the result of the $\caus$ construction in which it would be interesting to analyse constructions for one-way signalling and their logical role. For example, settings with infinite-dimensional systems are rarely compact closed, real quantum mechanics doesn't have enough causal states (since it is compact closed but doesn't admit local discrimination), and $\mathrm{Rel}$ fails Condition \ref{prec:split}. We needn't expect the results of this paper to generalise to other theories since (non-deterministic) coherence spaces form a $\bv$-category using a similar construction to $\mathbf{A \obslash B}$ for the non-commutative operator \cite{Blute2012}, but this is distinct from the naive adaptation of $\mathbf{A < B}$ (instead of a constant residual, every local effect must be a subset of some constant clique) which is not self-dual in this category.


\bigskip

\noindent \textbf{Acknowledgements.} The authors would like to thank Alessio Guglielmi for posing the question of modelling $\bv$-logic within $\caus[\catc]$, as well as Chris Barrett, Lutz Stra{\ss}burger, and members of the quantum group at University of Oxford for useful discussions. WS was supported by Cambridge Quantum (Quantinuum Ltd). AK acknowledges support from Grant No. 61466 from the John Templeton Foundation as part of the QISS project. The opinions expressed in this publication are those of the authors and do not necessarily reflect the views of the John Templeton Foundation.

\bibliography{library}
\bibliographystyle{plainurl}

\newpage

\appendix

\section{Subtractive Closure}\label{sec:subtraction}

Categories enriched with addition need not always have subtraction, but the additive precausal assumptions consider freely adjoining subtraction.

\begin{definition}
Given an additive precausal category $\catc$, the free subtractive closure $\sub(\catc)$ formed by the objects of $\catc$ and morphisms $A \to B$ are equivalence classes of pairs of morphisms $f, g \in \catc(A, B)$ under the relation $(f, g) \sim (f', g') \xLeftrightarrow{\mathrm{def}} f + g' = f' + g$. Composition of morphisms is defined as $(f, g) \fatsemi (x, y) \triangleq (f \fatsemi x + g \fatsemi y, f \fatsemi y + g \fatsemi x)$.
\end{definition}

The intent behind this construction is that a morphism $(f, g)$ should represent the expression $f-g$, and hence this should be seen as equivalent to $f'-g'$ when $(f - g) + g + g' = f + g' = f' + g = (f' - g') + g + g'$ by cancellativity.

\begin{proposition}
$\sub(\catc)$ is an $\mathrm{Ab}$-enriched category.
\end{proposition}

\begin{proof}
Firstly, we need $\sim$ to be an equivalence relation. Reflexivity and symmetry are both trivial. For transitivity, given $(f, g) \sim (a, b) \sim (x, y)$ we have $f + b = a + g$ and $a + y = x + b$. Hence $f + y + b = a + g + y = x + b + g$, so by cancellativity $f + y = x + g$, i.e. $(f, g) \sim (x, y)$.

Since morphisms of $\sub(\catc)$ are equivalence classes, we need composition to be well-defined with respect to our choices of representatives. That is, if $(f, g) \sim (f', g')$ and $(x, y) \sim (x', y')$, then $(f, g) \fatsemi (x, y) \sim (f', g') \fatsemi (x', y')$. From $f + g' = f' + g$ and $x + y' = x' + y$ we have:
\begin{equation}
\begin{split}
f \fatsemi x + g' \fatsemi x &= f' \fatsemi x + g \fatsemi x \\
f' \fatsemi y + g \fatsemi y &= f \fatsemi y + g' \fatsemi y \\
f' \fatsemi y' + g \fatsemi y' &= f \fatsemi y' + g' \fatsemi y' \\
f \fatsemi x' + g' \fatsemi x' &= f' \fatsemi x' + g \fatsemi x' \\
f \fatsemi x + f \fatsemi y' &= f \fatsemi x' + f \fatsemi y \\
g \fatsemi x' + g \fatsemi y &= g \fatsemi x + g \fatsemi y' \\
g' \fatsemi x' + g' \fatsemi y &= g' \fatsemi x + g' \fatsemi y' \\
f' \fatsemi x + f' \fatsemi y' &= f' \fatsemi x' + f' \fatsemi y \\
\end{split}
\end{equation}

We note that, since $\catc$ is compact closed, $\forall A, B, C, D . \forall f : A \to B . \forall g, h : C \to D . f \otimes (g + h) = f \otimes g + f \otimes h$. In particular, $\forall f . f + f = 2 \cdot f$ where $2$ is the scalar $\id_I + \id_I$. Summing these equations and using cancellativity:
\begin{equation}
2 \cdot (f \fatsemi x + g \fatsemi y + f' \fatsemi y' + g' \fatsemi x') = 2 \cdot (f \fatsemi y + g \fatsemi x + f' \fatsemi x' + g' \fatsemi y')
\end{equation}

By \ref{apc:scalars}, $2$ is invertible. The resulting equation exactly gives us $(f, g) \fatsemi (x, y) \sim (f', g') \fatsemi (x', y')$.

Next, composition must be unital. The identity on $A$ is $(\id_A, 0_{A,A})$ since:
\begin{equation}
\begin{split}
(f, g) \fatsemi (\id_B, 0_{B,B}) &= (f \fatsemi \id_B + g \fatsemi 0_{B,B}, f \fatsemi 0_{B,B} + g \fatsemi \id_B) \\
&= (f + 0_{A,B}, 0_{A,B} + g) \\
&= (f, g) \\
&= (\id_A \fatsemi f + 0_{A,A} \fatsemi g, \id_A \fatsemi g + 0_{A,A} \fatsemi f) \\
&= (\id_A, 0_{A,A}) \fatsemi (f, g)
\end{split}
\end{equation}

We also need composition to be associative. Given $(f, g) : A \to B$, $(x, y) : B \to C$, and $(u, v) : C \to D$:
\begin{equation}
\begin{split}
((f, g) \fatsemi (x, y)) \fatsemi (u, v) &= (f \fatsemi x + g \fatsemi y, f \fatsemi y + g \fatsemi x) \fatsemi (u, v) \\
&= \left( \begin{array}{l} f \fatsemi x \fatsemi u + g \fatsemi y \fatsemi u + f \fatsemi y \fatsemi v + g \fatsemi x \fatsemi v, \\ f \fatsemi y \fatsemi u + g \fatsemi x \fatsemi u + f \fatsemi x \fatsemi v + g \fatsemi y \fatsemi v \end{array} \right) \\
&= (f, g) \fatsemi (x \fatsemi u + y \fatsemi v, x \fatsemi v + y \fatsemi u) \\
&= (f, g) \fatsemi ((x, y) \fatsemi (u, v))
\end{split}
\end{equation}

Finally, $\mathrm{Ab}$-enrichment comes from inheriting the enrichment in commutative monoids from $\catc$ via $(f, g) + (x, y) \triangleq (f + x, g + y)$ and we have inverses $-(f, g) \triangleq (g, f)$. We can similarly show that these are well-defined regardless of the choice of representatives by showing $\sim$ is preserved. Invertibility comes from $(f, g) + - (f, g) = (f, g) + (g, f) = (f + g, g + f) \sim (0_{A,B}, 0_{A,B})$.
\end{proof}

\begin{proposition}
There is a faithful embedding $\fsub{-}: \catc \to \sub(\catc)$ that is the identity on objects and maps morphisms as $\fsub{f} := (f, 0_{A,B})$.
\end{proposition}

\begin{proof}
For functoriality, $(f, 0_{A,B}) \fatsemi (g, 0_{B,C}) \sim (f \fatsemi g, 0_{A,C})$ and identities are $(\id_A, 0_{A,A})$. The identity on objects is injective, so this functor is an embedding. For faithfulness, if we have $f, g : A \to B$ such that $(f, 0_{A,B}) \sim (g, 0_{A,B})$, then $f = f + 0_{A,B} = g + 0_{A,B} = g$.
\end{proof}

For the remainder, we will show that this embedding preserves much of the core categorical structure of $\catc$.

\begin{proposition}
The embedding $\fsub{-} : \catc \to \sub(\catc)$ is strong monoidal, where the product in $\sub(\catc)$ acts on objects identically to $\catc$ and $(f, g) \otimes (x, y) \triangleq (f \otimes x + g \otimes y, f \otimes y + g \otimes x)$.
\end{proposition}

\begin{proof}
Firstly, this monoidal product is well-defined, i.e. if $(f, g) \sim (f', g')$ and $(x, y) \sim (x', y')$ then $(f, g) \otimes (x, y) \sim (f', g') \otimes (x', y')$. This can be proved identically to the case for sequential composition.

We inherit the associators, unit, and unitors from $\catc$ as $(\alpha_{A,B,C}, 0_{(A \otimes B) \otimes C, A \otimes (B \otimes C)})$, $I$, $(\lambda_A, 0_{I \otimes A, A})$, $(\rho_A, 0_{A \otimes I, A})$. These will still be invertible by functoriality of the embedding. For naturality:
\begin{equation}
\begin{split}
& ((f,g) \otimes (x,y)) \otimes (u,v) \fatsemi (\alpha_{D,E,F}, 0) \\
={}& \left( \begin{array}{l} ((f \otimes x + g \otimes y) \otimes u + (f \otimes y + g \otimes x) \otimes v) \fatsemi \alpha_{D,E,F}, \\ ((f \otimes x + g \otimes y) \otimes v + (f \otimes y + g \otimes x) \otimes u) \fatsemi \alpha_{D,E,F} \end{array} \right) \\
={}& \left( \begin{array}{l} \alpha_{A,B,C} \fatsemi (f \otimes (x \otimes u + y \otimes v) + g \otimes (x \otimes v + y \otimes u)), \\ \alpha_{A,B,C} \fatsemi (f \otimes (x \otimes v + y \otimes u) + g \otimes (x \otimes u + y \otimes v)) \end{array} \right) \\
={}& (\alpha_{A,B,C}, 0) \fatsemi (f,g) \otimes ((x,y) \otimes (u,v))
\end{split}
\end{equation}

\begin{equation}
\begin{split}
& (f, g) \otimes (\id_I, 0_{I,I}) \fatsemi (\rho_B, 0_{B \otimes I,B}) \\
={}& (f \otimes \id_I + g \otimes 0_{I,I}, f \otimes 0_{I,I} + g \otimes \id_I) \fatsemi (\rho_B, 0_{B \otimes I,B}) \\
={}& (f \otimes \id_I, g \otimes \id_I) \fatsemi (\rho_B, 0_{B \otimes I,B}) \\
={}& ((f \otimes \id_I) \fatsemi \rho_B, (g \otimes \id_I) \fatsemi \rho_B) \\
={}& (\rho_A \fatsemi f, \rho_A \fatsemi g) \\
={}& (\rho_A, 0_{A \otimes I,A}) \fatsemi (f, g)
\end{split}
\end{equation}
and similarly for $(\lambda_A, 0_{I \otimes A, A})$. In the above, we have used compact closure of $\catc$ to infer that $f \otimes 0_{C,D} = 0_{A \otimes C, B \otimes D}$. We may also inherit the triangle and pentagon equations by functoriality of the embedding. We now have that $\sub(\catc)$ is a monoidal category.

For the embedding to be strong monoidal, we first note that the maps $\fsub{-}_0 : I^{\sub(\catc)} \xrightarrow{\simeq} \fsub{I^{\catc}}$ and $(\fsub{-}_2)_{A,B} : \fsub{A} \otimes \fsub{B} \xrightarrow{\simeq} \fsub{A \otimes B}$ are given by identities. Naturality of $\fsub{-}_2$ is just $(f, 0) \otimes (g, 0) = (f \otimes g + 0 \otimes 0, f \otimes 0 + 0 \otimes g) = (f \otimes g, 0)$. The compatibility equations reduce to $\alpha_{\fsub{A}, \fsub{B}, \fsub{C}} = \fsub{\alpha_{A,B,C}}$, $\lambda_{\fsub{A}} = \fsub{\lambda_A}$ and $\rho_{\fsub{A}} = \fsub{\rho_A}$.
\end{proof}

\begin{proposition}
$\sub(\catc)$ is compact closed, and the embedding from $\catc$ preserves symmetry, duals, cups, and caps.
\end{proposition}

\begin{proof}
The symmetry in $\sub(\catc)$ is $(\sigma_{A,B}, 0_{A \otimes B,B \otimes A}) = \fsub{\sigma_{A,B}}$. The hexagon equation and invertibility follow by functoriality. For naturality:
\begin{equation}
\begin{split}
& (f,g) \otimes (x,y) \fatsemi (\sigma_{C,D}, 0) \\
={}& (f \otimes x + g \otimes y, f \otimes y + g \otimes x) \fatsemi (\sigma_{C,D}, 0) \\
={}& ((f \otimes x) \fatsemi \sigma_{C,D} + (g \otimes y) \fatsemi \sigma_{C,D}, (f \otimes y) \fatsemi \sigma_{C,D} + (g \otimes x) \fatsemi \sigma_{C,D}) \\
={}& (\sigma_{A,B} \fatsemi (x \otimes f) + \sigma_{A,B} \fatsemi (y \otimes g), \sigma_{A,B} \fatsemi (x \otimes g) + \sigma_{A,B} \fatsemi (y \otimes f)) \\
={}& (\sigma_{A,B}, 0) \fatsemi (x \otimes f + y \otimes g, x \otimes g + y \otimes f) \\
={}& (\sigma_{A,B}, 0) \fatsemi (x,y) \otimes (f,g)
\end{split}
\end{equation}

We appeal to the standard result that monoidal functors preserve duals \cite[Theorem 3.14]{Heunen2019}. The given cup and cap are $\fsub{\eta_A} = \fsub{-}_0 \fatsemi \fsub{\eta_A} \fatsemi (\fsub{-}_2)_{A^*,A}^{-1} : I \to \fsub{A^*} \otimes \fsub{A}$ and $\fsub{\epsilon_A} = (\fsub{-}_2)_{A^*,A} \fatsemi \fsub{\epsilon_A} \fatsemi \fsub{-}_0^{-1} : \fsub{A} \otimes \fsub{A^*} \to I$. Since the embedding is bijective on objects, all objects in $\sub(\catc)$ have a dual determined by their dual in $\catc$, hence it is compact closed.
\end{proof}

\begin{proposition}
$\sub(\catc)$ has biproducts, and the embedding from $\catc$ preserves the biproduct structure.
\end{proposition}

\begin{proof}
We can inherit the biproduct $A \oplus B$ with injections $(\iota_A, 0_{A,A \oplus B}), (\iota_B, 0_{A,A \oplus B})$ and projections $(p_A, 0_{A \oplus B,A}), (p_B, 0_{A \oplus B,B})$. The characteristic equations $\iota_{A_i} \fatsemi p_{A_j} = \begin{cases}
\id_{A_i} & i = j \\
0_{A_i,A_j} & i \neq j
\end{cases}$ and $p_A \fatsemi \iota_A + p_B \fatsemi \iota_B = \id_{A \oplus B}$ are preserved by functoriality.
\end{proof}

\begin{proposition}\label{prop:sub_preserves_basis}
Any set of states $\{\rho_i\}_i \subseteq \catc(I,A)$ is a basis for $A$ in $\catc$ (in the sense of being a minimal set such that $(\forall i . \rho_i \fatsemi f = \rho_i \fatsemi g) \Rightarrow f = g$) iff $\{\fsub{\rho_i}\}_i$ is a basis for $A$ in $\sub(\catc)$.
\end{proposition}

\begin{proof}
It is sufficient to just show equivalence between the ability to distinguish morphisms in each category, since preservation of minimality follows as a basic consequence of this (any subset of the $\rho_i$ would distinguish all morphisms in $\catc$ iff it does so in $\sub(\catc)$).

$\Rightarrow$: Suppose $\forall B . \forall f, g \in \catc(A,B) . (\forall i . \rho_i \fatsemi f = \rho_i \fatsemi g) \Rightarrow f = g$. Consider an arbitrary $B$ and some $(f^+, f^-), (g^+, g^-) \in \sub(\catc)(A,B)$. Suppose that $\forall i . \fsub{\rho_i} \fatsemi (f^+, f^-) \sim \fsub{\rho_i} \fatsemi (g^+, g^-)$. Unpacking this, we have $\forall i . \rho_i \fatsemi (f^+ + g^-) = \rho_i \fatsemi (g^+ + f^-)$. Since the $\rho_i$ form a basis in $\catc$, we have $f^+ + g^- = g^+ + f^-$, i.e. $(f^+, f^-) \sim (g^+, g^-)$.

$\Leftarrow$: Suppose $\forall B . \forall (f^+, f^-), (g^+, g^-) \in \sub(\catc)(A,B) . (\forall i . \fsub{\rho_i} \fatsemi (f^+, f^-) \sim \fsub{\rho_i} \fatsemi (g^+, g^-)) \Rightarrow (f^+, f^-) \sim (g^+, g^-)$. Consider an arbitrary $B$ and some $f, g \in \catc(A,B)$ such that $\forall i . \rho_i \fatsemi f = \rho_i \fatsemi g$. By adding $\rho_i \fatsemi 0_{A,B}$ on both sides we get $\forall i . \fsub{\rho_i} \fatsemi \fsub{f} \sim \fsub{\rho_i} \fatsemi \fsub{g}$, and then we can use the basis property to derive that $\fsub{f} \sim \fsub{g}$, i.e. $f = f + 0_{A,B} = g + 0_{A,B} = g$.
\end{proof}

\section{Proofs for Section \ref{sec:apc}}

\begin{proposition}[Restatement of Proposition \ref{prop:apc-cancellative}]
  In an additive precausal category:
  \begin{equation}
    \tag{APC5a}
    \hfill
    \forall f, g, h \in \catc(A,B).\ 
    f + h = g + h \Rightarrow f = g
    \hfill
  \end{equation}
\end{proposition}

\begin{proof}
  By \ref{apc:basis}, we can fix bases of causal states $\{\rho_i\}_i$, $\{\mu_j\}_j$ for the systems $A$ and $B^*$ respectively. By transposing the $\mu_j$, we can regard them as effects $\mu_j^* : B \to I$ defined as $\mu_j^* := (\mu_j \otimes \id_B) \fatsemi \epsilon_{B^*}$. Now, suppose $f + h = g + h$. Then, for all $i,j$:
  \[
    \hfill
    \rho_i \fatsemi (f + h) \fatsemi \mu_j^* = 
    \rho_i \fatsemi (g + h) \fatsemi \mu_j^*
    \ \ \implies\ \ 
    \rho_i \fatsemi f \fatsemi \mu_j^* +
    \rho_i \fatsemi h \fatsemi \mu_j^* =
    \rho_i \fatsemi g \fatsemi \mu_j^* +
    \rho_i \fatsemi h \fatsemi \mu_j^*
    \hfill
  \]
  Applying cancellation for scalars \ref{apc:scalars}, we conclude that $\rho_i \fatsemi f \fatsemi \mu_j^* = \rho_i \fatsemi g \fatsemi \mu_j^*$. Applying \ref{apc:basis} to the basis $\{\mu_j\}_j$ gives $\rho_i \fatsemi f = \rho_i \fatsemi g$. Then, a second application of \ref{apc:basis} to $\{\rho_i\}_i$ gives $f = g$.
\end{proof}

\begin{proposition}[Restatment of Proposition \ref{prop:subc-field}]
  For a precausal category $\catc$, the scalars $K := \sub(\catc)(I,I)$ are a field, and hence $\sub(\catc)$ is enriched over $K$-vector spaces.
\end{proposition}

\begin{proof}
  As the subtractive closure of the semiring $\catc(I,I)$, we already know that $K$ is a ring, so it suffices to show that $K$ has multiplicative inverses. Take a non-zero element $k \in K$, represented by a pair $(x,y)$ of elements $x,y \in \catc(I,I)$. Then, since $\catc(I,I)$ is totally pre-ordered (\ref{apc:scalars}), there exists $z$ such that $x = y + z$ or $x + z = y$. In the first case, $(x,y) \sim (z,0)$ and in the second, $(x,y) \sim (0,z)$. Since at least one of $x$ and $y$ is non-zero, $z$ must also be non-zero, so by \ref{apc:scalars} it has an inverse $z^{-1}$. Hence we can take $k^{-1}$ to be either $(z^{-1},0)$ or $(0,z^{-1})$. Enrichment in $K$-vector spaces then follows immediately.
\end{proof}

\begin{proposition}[Restatment of Proposition \ref{prop:bintests}]
  For any $f: A \to B$ in an additive precausal category, there exists $f' : A \to B$ and a scalar $\lambda$ such that:
  \begin{equation}\label{apc:bintests}
    \tag{APC5a}
    \hfill
    f + f' = \lambda \cdot \discard_A \fatsemi \maxmix_B
    \hfill
  \end{equation}
\end{proposition}

\begin{proof}
  We can first make $f$ into an effect $\pi_f : A^* \otimes B \to I$ using the compact structure: $\pi_f := (\id_{B^*} \otimes f) \fatsemi \epsilon_B$. Then, applying \ref{apc:effects}, we get a complement $\pi'$ satisfying $\pi_f + \pi' = \lambda \cdot (\discard_{B^*} \otimes \discard_A)$. Then, taking $f' := (\eta_B \otimes \id_A) \fatsemi (\id_B \otimes \pi')$ satisfies the required property \ref{apc:bintests}.
\end{proof}

\begin{lemma}[Restatement of Lemma \ref{lemma:extend_basis}]\label{proof:lemma:extend_basis}
Given any set of morphisms in $\catc(A,B)$ that are linearly independent in $\sub(\catc)$, they can be exteded to a basis in $\catc$ with a dual basis in $\sub(\catc)$.
\end{lemma}

\begin{proof}
By compact closure it is sufficient to consider just states $\catc(I,A)$.

Let $\{\rho_i\}_i \subseteq \catc(I,A)$ be a minimal finite basis which must exist wlog from \ref{apc:basis}. We will start by showing that this has a dual basis, i.e. a set of effects $\{e_i\}_i \subseteq \sub(\catc)(A,I)$ such that $\fsub{\rho_i} \fatsemi e_j \sim \delta_{i,j} = \begin{cases}
\id_I & i=j \\
0_{I,I} & i\neq j
\end{cases}$. The vector of scalars $\rho_i \fatsemi \pi$ must uniquely identify any effect $\pi \in \catc(A,I)$, giving us a coordinate system for effects. Similarly, the set of effects $\{\pi_j\}_j \subseteq \catc(A,I)$ formed as the transpose of a minimal basis for $A^*$ from \ref{apc:basis} can yield coordinates $\rho \fatsemi \pi_j$ that uniquely describe any state $\rho \in \catc(I,A)$.

We can build a matrix of the scalars (in $\sub(\catc)$) formed by the inner products $m_{i,j} = \fsub{\rho_i \fatsemi \pi_j}$, so the rows describe coordinates of the states $\{\rho_i\}_i$ and columns for the effects $\{\pi_j\}_j$. Performing column operations such as rescaling by a non-zero scalar $\alpha$ or summing columns $j$ and $k$ generates the coordinates of the effect $\alpha \cdot \pi_j$ or $\pi_j + \pi_k$. By cancellative addition and invertibility of non-zero scalars from \ref{apc:scalars}, we can still represent each effect in $\{\pi_j\}_j$ as a linear combination of the new column effects. Since the scalars of $\sub(\catc)$ form field, we can apply Gaussian elimination with column operations to yield column effects $\{e_j\}_j \subseteq \sub(\catc)(A,I)$ such that $\fsub{\rho_i} \fatsemi e_j = \delta_{i,j}$ (Gaussian elimination completes since any zero columns would indicate $\{\pi_j\}_j$ were linearly dependent and not a minimal basis and extra unsolved rows would similarly contradict minimality of $\{\rho_i\}_i$).




The new column effects $\{e_j\}_j$ still represent a basis for $A^*$ in $\sub(\catc)$ under transposition. The $\{\fsub{\pi_j^*}\}_j$ form a basis for $A^*$ in $\sub(\catc)$ by \ref{prop:sub_preserves_basis}, and one output of Gaussian elimination (with back substitution) is a matrix of constants $\{\alpha_{j,k}\}_{j,k} \subseteq \sub(\catc)(I,I)$ such that $\fsub{\pi_j} \sim \sum_k \alpha_{j,k} \cdot e_j$. Hence if any two morphisms $f, g \in \sub(\catc)(A^*,B)$ agree on $\{e_j^*\}_j$, they must agree on any linear combination, meaning they agree on each of $\{\fsub{\pi_j^*}\}_j$ and so they are equal. Minimality comes from the need to distinguish the $\rho_i$ under transposition (if some $e_i^*$ were excluded, the remainder whould not be able to determine $\fsub{\rho_i^*} \not\sim \fsub{0_{A^*,I}} \in \sub(\catc)(A^*,I)$ since all inner products will be zero in both cases). We now have that $\{e_j\}_j$ is a dual basis to $\{\rho_i\}_i$.

Given any set of states $\{s_k\}_k \subseteq \catc(I,A)$ that are linearly independent in $\sub(\catc)$, we can extend the set to $\{s_i'\}_i$ by adding any terms from $\{\rho_i\}_i$ that preserve linear independence. We can then similarly diagonalise the inner products with $\{e_j\}_j$ to represent each $\rho_i$ as a linear combination of $\{s_i'\}_i$ and prove that they form basis and construct a dual basis for it.
\end{proof}

\begin{theorem}[Restatement of Theorem \ref{thrm:affine}]\label{proof:thrm:affine}
Given any flat set $c \subseteq \catc(I,A)$ for a non-zero $A$, $c^{**} = \affp(c)$.
\end{theorem}

\begin{proof}
$\supseteq$: Suppose $s \in \catc(I,A)$ can be expressed as an affine combination of elements of $c$, i.e. $s + \sum_i \alpha_i \cdot \rho_i^- = \sum_j \beta_j \cdot \rho_j^+$ for $\rho_i^-, \rho_j^+ \in c$ and $1 + \sum_i \alpha_i = \sum_j \beta_j$. For any effect $\pi \in c^*$ we have:
\begin{equation}
\begin{split}
s \fatsemi \pi + \sum_i \alpha_i &= \left(s + \sum_i \alpha_i \cdot \rho_i^-\right) \fatsemi \pi \\
&= \left(\sum_j \beta_j \cdot \rho_j^+\right) \fatsemi \pi \\
&= \sum_j \beta_j \\
&= 1 + \sum_i \alpha_i
\end{split}
\end{equation}
hence by cancellativity $s \fatsemi \pi = 1$, so $s \in c^{**}$.

$\subseteq$: We start by fixing any maximal linearly independent subset of $c$ and extend it to a basis $\{\rho_i\}_i \subseteq \catc(I,A)$ with a dual basis $\{e_i\}_i \subseteq \sub(\catc)(A,I)$ using Lemma \ref{lemma:extend_basis}. Let $\mathcal{I} = \{i | \rho_i \in c\}$ index the subset of basis elements in $c$. By resolving the identity, we can always express any $s \in \catc(I,A)$ as $\fsub{s} \sim \sum_i (\fsub{s} \fatsemi e_i) \cdot \fsub{\rho_i}$.

Since our construction of the basis used a maximal linearly independent subset of $c$, each term in $c$ must be expressible as a linear combination $\sum_{i \in \mathcal{I}} \alpha_i \cdot \rho_i$ and therefore $\forall i \notin \mathcal{I} . \forall t \in c . \fsub{t} \fatsemi e_i \sim 0$. If we choose some representative decomposition $e_i \sim \fsub{e_i^+} - \fsub{e_i^-}$ for $e_i^+, e_i^- \in \catc(A,I)$, we have $\forall i \notin \mathcal{I} . \forall t \in c . t \fatsemi e_i^+ = t \fatsemi e_i^-$.

Since $c$ is flat, there exists some invertible $\mu \in \catc(I,I)$ such that $\mu \cdot \discard_A \in c^*$. By \ref{apc:effects}, there is also some invertible $\lambda \in \catc(I,I)$ and $e' \in \catc(A,I)$ such that $e' + e_i^- = \lambda \cdot \discard_A$. We can show that $\mu \cdot \lambda^{-1} \cdot (e' + e_i^+) \in c^*$ for any $i \notin \mathcal{I}$ since for any $t \in c$:
\begin{equation}
\begin{split}
t \fatsemi \mu \cdot \lambda^{-1} (e' + e_i^+) &= \mu \cdot \lambda^{-1} \cdot \left(t \fatsemi e' + t \fatsemi e_i^+\right) \\
&= \mu \cdot \lambda^{-1} \cdot \left(t \fatsemi e' + t \fatsemi e_i^-\right) \\
&= \mu \cdot \lambda^{-1} \cdot \left(t \fatsemi \lambda \cdot \discard_A\right) \\
&= t \fatsemi \left(\mu \cdot \discard_A\right) \\
&= 1
\end{split}
\end{equation}


By considering $\mu \cdot \discard_A, \mu \cdot \lambda^{-1} \cdot (e' + e_i^+) \in c^*$ and an arbitrary $s \in c^{**}$,
\begin{equation}
\begin{split}
\mu \cdot \lambda^{-1} \cdot s \fatsemi e_i^- + 1 &= \mu \cdot \lambda^{-1} \cdot s \fatsemi \left(e' + e_i^+ + e_i^-\right) \\
&= \mu \cdot \lambda^{-1} \cdot s \fatsemi \left(\lambda \cdot \discard_A + e_i^+\right) \\
&= \mu \cdot \lambda^{-1} \cdot s \fatsemi e_i^+ + s \fatsemi \left(\mu \cdot \discard_A\right) \\
&= \mu \cdot \lambda^{-1} \cdot s \fatsemi e_i^+ + 1
\end{split}
\end{equation}

By cancellativity and invertibility of $\mu$ and $\lambda$, we have $s \fatsemi e_i^- = s \fatsemi e_i^+$.

Expanding $s$, we now have $\fsub{s} \sim \sum_i (\fsub{s} \fatsemi e_i) \cdot \fsub{\rho_i} \sim \sum_{i \in \mathcal{I}} (\fsub{s} \fatsemi e_i) \cdot \fsub{\rho_i}$, i.e. $s$ is a linear combination of terms in $c$. This combination is affine since:
\begin{equation}
\begin{split}
\sum_{i \in \mathcal{I}} \fsub{s} \fatsemi e_i &\sim \sum_{i \in \mathcal{I}} \left(\fsub{s} \fatsemi e_i\right) \cdot \left(\fsub{\rho_i} \fatsemi \mu \cdot \fsub{\discard_A}\right) \\
&\sim \left(\sum_{i \in \mathcal{I}} (\fsub{s} \fatsemi e_i) \cdot \fsub{\rho_i} \right) \fatsemi \mu \cdot \fsub{\discard_A} \\
&\sim \fsub{s} \fatsemi \mu \cdot \fsub{\discard_A} \\
&\sim 1
\end{split}
\end{equation}
\end{proof}

\begin{theorem}[Restatement of Theorem \ref{thrm:no_interaction_with_trivial}]\label{proof:thrm:no_interaction_with_trivial}
If $\mathbf{A}$ is $c_\mathbf{A} = \{\mu \cdot \maxmix_A\}$ for any non-zero $A$, then every $h \in c_{\mathbf{A \otimes B}}$ is a product morphism of the form $\mu \cdot \maxmix_A \otimes g$ for some $g \in c_\mathbf{B}$.
\end{theorem}

\begin{proof}
By Theorem \ref{thrm:affine}, we can represent any such $h$ as an affine combination of product terms:
\begin{equation}
\begin{split}
h + \sum_{i \in \mathcal{I}^-} \alpha_i^- \cdot \left(\mu \cdot \maxmix_A\right) \otimes g_i^- &= \sum_{i \in \mathcal{I}^+} \alpha_i^+ \cdot \left(\mu \cdot \maxmix_A\right) \otimes g_i^+ \\
h + \left(\mu \cdot \maxmix_A\right) \otimes \left(\sum_{i \in \mathcal{I}^-} \alpha_i^- \cdot g_i^-\right) &= \left(\mu \cdot \maxmix_A\right) \otimes \left(\sum_{i \in \mathcal{I}^+} \alpha_i^+ \cdot g_i^+\right)
\end{split}
\end{equation}
where $1 + \sum_{i \in \mathcal{I}^-} \alpha_i^- = \sum_{i \in \mathcal{I}^+} \alpha_i^+$ and $\{g_i^\pm\}_{i \in \mathcal{I}^\pm} \subseteq c_\mathbf{B}$. Since $A$ is non-zero, $c_\mathbf{A}^*$ is non-empty, so we can consider applying some effect $\pi \in c_\mathbf{A}^*$ to reduce the above equation to:
\begin{equation}
h \fatsemi \left(\pi \otimes \id_B\right) + \sum_{i \in \mathcal{I}^-} \alpha_i^- \cdot g_i^- = \sum_{i \in \mathcal{I}^+} \alpha_i^+ \cdot g_i^+
\end{equation}
with $h \fatsemi (\pi \otimes \id_B) \in c_\mathbf{B}$. Combining the decompositions so far gives:
\begin{equation}
\begin{split}
h + \left(\mu \cdot \maxmix_A\right) \otimes \left(\sum_{i \in \mathcal{I}^-} \alpha_i^- \cdot g_i^-\right) &= \left(\mu \cdot \maxmix_A\right) \otimes \left(\sum_{i \in \mathcal{I}^+} \alpha_i^+ \cdot g_i^+\right) \\
&= \left(\mu \cdot \maxmix_A\right) \otimes \left(h \fatsemi \left(\pi \otimes \id_B\right)\right) + \left(\mu \cdot \maxmix_A\right) \otimes \left(\sum_{i \in \mathcal{I}^-} \alpha_i^- \cdot g_i^-\right)
\end{split}
\end{equation}

So by cancellativity of addition (Proposition \ref{prop:apc-cancellative}) we have $h = \left(\mu \cdot \maxmix_A\right) \otimes (h \fatsemi (\pi \otimes \id_B))$.
\end{proof}

\section{Proofs for Section \ref{sec:additive}}

\begin{lemma}[Restatement of Lemma \ref{lemma:additive_defs_equivalent}]\label{proof:lemma:additive_defs_equivalent}
The alternative definitions of $c_{\mathbf{A \times B}}$ and $c_{\mathbf{A \oplus B}}$ are equivalent; that is, Equations \ref{eq:def_product} and \ref{eq:def_coproduct} hold.
\end{lemma}

\begin{proof}
\ref{eq:def_product} $\subseteq$: For any $\rho \in \left( \{p_A \fatsemi \pi_A | \pi_A \in c_\mathbf{A}^* \} \cup \{p_B \fatsemi \pi_B | \pi_B \in c_\mathbf{B}^* \} \right)^*$, the $\eta$-rule for products allows us to expand it as $\rho = (\rho \fatsemi p_A, \rho \fatsemi p_B)$. Then $\rho \fatsemi p_A \in c_\mathbf{A}$ since $\rho$ satisfies $\rho \fatsemi p_A \fatsemi \pi_A = 1$ for all $\pi_A \in c_\mathbf{A}^*$, and similarly $\rho \fatsemi p_B \in c_\mathbf{B}$.

\ref{eq:def_product} $\supseteq$: For any states $\rho_A \in c_\mathbf{A}, \rho_B \in c_\mathbf{B}$ and effects $\pi_A \in c_\mathbf{A}^*, \pi_B \in c_\mathbf{B}^*$, we have $(\rho_A, \rho_B) \fatsemi p_A \fatsemi \pi_A = \rho_A \fatsemi \pi_A = 1$ and $(\rho_A, \rho_B) \fatsemi p_B \fatsemi \pi_B = \rho_B \fatsemi \pi_B = 1$.

By swapping the roles of states/products/projections with effects/coproducts/injections, we get $\left( \{ \rho_A \fatsemi \iota_A | \rho_A \in c_\mathbf{A} \} \cup \{ \rho_B \fatsemi \iota_B | \rho_B \in c_\mathbf{B} \} \right)^* = \left\{ \langle \pi_A, \pi_B \rangle | \pi_A \in c_\mathbf{A}^*, \pi_B \in c_\mathbf{B}^* \right\}$, from which we can take the dual of both sides to obtain Equation \ref{eq:def_coproduct}.
\end{proof}

\begin{corollary}[Restatement of Corollary \ref{corollary:additive_demorgan}]\label{proof:corollary:additive_demorgan}
The operators $\times$ and $\oplus$ are De Morgan duals under $(-)^*$.
\end{corollary}

\begin{proof}
This is immediate from the symmetric definitions of $c_{\mathbf{A \times B}}$ and $c_{\mathbf{A \oplus B}}$. For example,
\begin{equation}
\begin{split}
c_{\mathbf{A \oplus B}}^* &= \{\langle \pi_A, \pi_B \rangle | \pi_A \in c_\mathbf{A}^* \subseteq \catc(A,I), \pi_B \in c_\mathbf{B}^* \subseteq \catc(B,I)\} \\
&= \{\langle \rho_A^*, \rho_B^* \rangle | \rho_A \in c_{\mathbf{A^*}} \subseteq \catc(I,A^*), \rho_B \in c_{\mathbf{B^*}} \subseteq \catc(I,B^*)\} \\
&= \{ p_A \fatsemi \rho_A^* + p_B \fatsemi \rho_B^* | \rho_A \in c_{\mathbf{A^*}}, \rho_B \in c_{\mathbf{B^*}}\} \\
&= \{ (\rho_A \fatsemi \iota_{A^*} + \rho_B \fatsemi \iota_{B^*})^* | \rho_A \in c_{\mathbf{A^*}}, \rho_B \in c_{\mathbf{B^*}}\} \\
&= \{ (\rho_A, \rho_B)^* | \rho_A \in c_{\mathbf{A^*}}, \rho_B \in c_{\mathbf{B^*}}\} \\
&= c_{\mathbf{A^* \times B^*}} \subseteq \catc(A \oplus B,I)
\end{split}
\end{equation}

Here we have used that injections and projections may be treated as transposes of one another when we suppose that the compact structure for the duality $A \oplus B \dashv A^* \oplus B^*$ in $\catc$ is built from those of $A \dashv A^*$ and $B \dashv B^*$ in the canonical way \cite{Heunen2019}.



\end{proof}

\begin{proposition}[Restatement of Proposition \ref{prop:product}]\label{proof:prop:product}
$\mathbf{A \times B}$ is a categorical product in $\caus[\catc]$.
\end{proposition}

\begin{proof}
Suppose we are given some $f : \mathbf{C \to A}$ and $g : \mathbf{C \to B}$. The existence and uniqueness of $(f, g)$ can be inherited from the fact that $A \oplus B$ is a (bi)product in $\catc$, so we just need to show that $(f, g)$, $p_A$, and $p_B$ are all causal.

Given any state $\rho \in c_\mathbf{C}$, since $f$ and $g$ are causal we have $\rho \fatsemi f \in c_\mathbf{A}$ and $\rho \fatsemi g \in c_\mathbf{B}$. The product definition of $c_{\mathbf{A \times B}}$ now gives that $\rho \fatsemi (f, g) = (\rho \fatsemi f, \rho \fatsemi g) \in c_{\mathbf{A \times B}}$, so $(f, g) : \mathbf{C \to A \times B}$.

For the projectors, we know that all states of $\mathbf{A \times B}$ are of the form $(\rho_A, \rho_B)$ for some $\rho_A \in c_\mathbf{A}$ and $\rho_B \in c_\mathbf{B}$. Since $(\rho_A, \rho_B) \fatsemi p_A = \rho_A \in c_\mathbf{A}$, we have $p_A : \mathbf{A \times B \to A}$ and similarly $p_B : \mathbf{A \times B \to B}$.
\end{proof}

\begin{proposition}[Restatement of Proposition \ref{prop:coproduct}]\label{proof:prop:coproduct}
$\mathbf{A \oplus B}$ is a categorical coproduct in $\caus[\catc]$.
\end{proposition}

\begin{proof}
Again, this can be obtained by dualising the proof from Proposition \ref{prop:product} to show that the injections and any $\langle f, g \rangle : A \oplus B \to C$ for $f : \mathbf{A \to C}$ and $g : \mathbf{B \to C}$ are all causal.
\end{proof}

\begin{proposition}\label{proof:prop:initial_terminal}[Restatement of Proposition \ref{prop:initial_terminal}]
The initial object in $\caus[\catc]$ is $\mathbf{0} := (0, \emptyset)$ and the terminal object is $\mathbf{1} := (0, \{0_{I,0}\})$. Furthermore, they are duals of each other and are units for $\mathbf{\oplus}$ and $\mathbf{\times}$ respectively.
\end{proposition}

\begin{proof}
Given any $\mathbf{A}$, there is a unique morphism $0_{0,A} \in \catc(0,A)$. The normalisation condition for $0_{0,A} : \mathbf{0 \to A}$ holds vacuously since there are no states of type $\mathbf{0}$ to preserve normalisation of. This is hence a unique morphism in $\caus[\catc]$, making $\mathbf{0}$ initial. An initial object is always a unit for coproducts.

Conversely, for any $\mathbf{A}$ there is a unique morphism $0_{A,0} \in \catc(A,0)$. For any $\rho \in c_\mathbf{A}$, $\rho \fatsemi 0_{A,0} = 0_{I,0} \in c_\mathbf{1}$ by terminality of the zero object, so $0_{A,0} : \mathbf{A \to 1}$. This similarly makes $\mathbf{1}$ terminal in $\caus[\catc]$, and a terminal object is always a unit for products.

For the duality, we note that the zero object is self-dual in any compact closed category $\catc$. $\emptyset^*$ will include the full homset of effects since the normalisation condition will vacuously hold. There is only one effect from the zero object so $\mathbf{0^*} = \mathbf{1}$. Conversely, $\{0_{I,0}\}^*$ will include any effect which composes with $0_{I,0}$ to give $\id_I$. The only effect to consider is $0_{0,I}$, and $0_{I,0} \fatsemi 0_{0,I} = 0_{I,I}$. $0_{I,I} \neq \id_I$ since the subtractive closure of scalars is a field (Proposition \ref{prop:subc-field}) requiring the zero and unit to be distinct so no normalised effect exists, meaning $\mathbf{1^*} = \mathbf{0}$.
\end{proof}

\begin{proposition}[Restatement of Proposition \ref{prop:coproduct_first_order}]\label{proof:prop:coproduct_first_order}
If $\mathbf{A}$ and $\mathbf{B}$ are both first-order types, then so is $\mathbf{A \oplus B}$.
\end{proposition}

\begin{proof}
If $c_\mathbf{A} = \left\{\discard_A\right\}^*$ and $c_\mathbf{B} = \left\{\discard_B\right\}^*$, then we have $c_{\mathbf{A \oplus B}} = \left\{\langle \pi_A, \pi_B \rangle | \pi_A \in c_\mathbf{A}^*, \pi_B \in c_\mathbf{B}^*\right\}^* = \left\{\langle \discard_A, \discard_B \rangle\right\}^* = \left\{\discard_{A \oplus B}\right\}^*$ using \ref{apc:discard}.
\end{proof}

\begin{proposition}[Restatement of Proposition \ref{prop:product_not_first_order}]\label{proof:prop:product_not_first_order}
If $A$ and $B$ are non-zero, then $\mathbf{A \times B}$ is never a first-order type.
\end{proposition}

\begin{proof}
By flatness, both $p_A \fatsemi \discard_A$ and $p_B \fatsemi \discard_B$ are in $c_{\mathbf{A \times B}}$ up to some invertible scalars. They are distinct as morphisms of $\catc$ since they can be distinguished using $(\rho_A, 0_{I,B})$ for any $\rho_A \in c_\mathbf{A}$ (projecting on $B$ will give zero whereas projecting on $A$ will give a non-zero scalar).
\end{proof}

\section{Proofs for Section \ref{sec:oneway}}

Many proofs will have to consider zero systems as a special case, since the causal types will either have no states or no effects to choose from. These cases are often relatively simple since the setting typically degenerates. For example, each of the following isomorphisms can be proved straightforwardly where the lack of states/effects generally manifests as a vacuous proof or a failure to find a witness (where $\mathbf{C}$ represents an arbitrary causal type with both causal states and effects).
\begin{align}\label{eq:zero_tensor}
\mathbf{0 \otimes 0} &\simeq \mathbf{0} &
\mathbf{0 \parr 0} &\simeq \mathbf{0} \nonumber \\
\mathbf{0 \otimes 1} &\simeq \mathbf{0} &
\mathbf{0 \parr 1} &\simeq \mathbf{1} \nonumber \\
\mathbf{1 \otimes 1} &\simeq \mathbf{1} &
\mathbf{1 \parr 1} &\simeq \mathbf{1} \\
\mathbf{0 \otimes C} &\simeq \mathbf{0} &
\mathbf{0 \parr C} &\simeq \mathbf{0} \nonumber \\
\mathbf{1 \otimes C} &\simeq \mathbf{1} &
\mathbf{1 \parr C} &\simeq \mathbf{1} \nonumber
\end{align}

These isomorphisms should not be surprising since they can also be proved from the rules of linear logic. The only ones that can't be mapped to valid sequents immediately are $\mathbf{0 \parr C} \to \mathbf{0}$ which requires $! (\sim C), \bot \parr C \vdash \bot$, and $\mathbf{1} \to \mathbf{1 \otimes C}$ which requires $!C, \top \vdash \top \otimes C$, i.e. we are encoding the existence of states and effects for $\mathbf{C}$ as explicit information we may use in the proof.

As for $\mathbf{<}$, $\mathbf{\triangleleft}$, and $\mathbf{\obslash}$, the proofs for characterising interaction with zero objects are still degenerate in the same way. For the $\mathbf{\triangleleft}$ decompositions, the intermediate system can be chosen as $\mathbf{Z} = \mathbf{I}$ in each case apart from $\mathbf{1 \triangleleft 0} \simeq \mathbf{1}$ where we choose $\mathbf{Z} = \mathbf{0}$ (which is first order since $c_\mathbf{0}^* = \{0_{0,I}\} = \{\discard_0\}$).
\begin{equation}\label{eq:zero_one_way}
\begin{array}{ccccccc}
\mathbf{0 < 0} &= &\mathbf{0 \triangleleft 0} &= &\mathbf{0 \obslash 0} &\simeq &\mathbf{0} \\
\mathbf{0 < 1} &= &\mathbf{0 \triangleleft 1} &= &\mathbf{0 \obslash 1} &\simeq &\mathbf{0} \\
\mathbf{1 < 0} &= &\mathbf{1 \triangleleft 0} &= &\mathbf{1 \obslash 0} &\simeq &\mathbf{1} \\
\mathbf{1 < 1} &= &\mathbf{1 \triangleleft 1} &= &\mathbf{1 \obslash 1} &\simeq &\mathbf{1} \\
\mathbf{0 < C} &= &\mathbf{0 \triangleleft C} &= &\mathbf{0 \obslash C} &\simeq &\mathbf{0} \\
\mathbf{1 < C} &= &\mathbf{1 \triangleleft C} &= &\mathbf{1 \obslash C} &\simeq &\mathbf{1} \\
\mathbf{C < 0} &= &\mathbf{C \triangleleft 0} &= &\mathbf{C \obslash 0} &\simeq &\mathbf{0} \\
\mathbf{C < 1} &= &\mathbf{C \triangleleft 1} &= &\mathbf{C \obslash 1} &\simeq &\mathbf{1}
\end{array}
\end{equation}


For the proofs below, we will just show the proofs for the cases where $A$ and $B$ are non-zero systems.

\begin{lemma}\label{lemma:one_way_equals_pcs}
$c_\mathbf{A} \obslash c_\mathbf{B} = c_\mathbf{A} < c_\mathbf{B}$.
\end{lemma}

\begin{proof}
$\subseteq$: Consider an arbitrary $h \in c_\mathbf{A} \obslash c_\mathbf{B}$ with corresponding witnesses $\mathcal{I}, \{f_i\}_{i \in \mathcal{I}}, f, \{g_i\}_{i \in \mathcal{I}}$, and any effect $\pi \in c_\mathbf{B}^*$. We have $g_i \fatsemi \pi = 1$ for each $g_i$ because $g_i \in c_\mathbf{B}$.

\begin{equation}
\begin{split}
\fsub{h \fatsemi \left(\id_A \otimes \pi\right)} &\sim \left(\sum_{i \in \mathcal{I}} f_i \otimes \fsub{g_i}\right) \fatsemi \fsub{\id_A \otimes \pi} \\
&\sim \sum_{i \in \mathcal{I}} \fsub{g_i \fatsemi \pi} \cdot f_i \\
&\sim \sum_{i \in \mathcal{I}} f_i \\
&\sim \fsub{f}
\end{split}
\end{equation}

By faithfulness of the embedding into $\sub(\catc)$, we have that $h \fatsemi (\id_A \otimes \pi) = f \in c_\mathbf{A}$ for any $\pi \in c_\mathbf{B}^*$ and hence $h \in c_\mathbf{A} < c_\mathbf{B}$.

$\supseteq$: Consider an arbitrary $h \in c_\mathbf{A} < c_\mathbf{B}$ with residual $m \in c_\mathbf{A}$. Let $\{\rho_i\}_i$ be a basis for $A$ in $\catc$ with a dual basis $\{e_i\}_i$ in $\sub(\catc)$ such that $\fsub{m} \fatsemi e_i \not\sim 0$ (and therefore has an inverse) for all $i$. We can guarantee this without loss of generality: we start with the default basis for $A$ from \ref{apc:basis} and find a dual basis using Lemma \ref{lemma:extend_basis}. Since $m$ is not the zero process, there must be some non-zero value $\fsub{m} \fatsemi e_a$. For any $i$ such that $\fsub{m} \fatsemi e_i \sim 0$, we can update $\rho_a' := \rho_a + \rho_i$ and $e_i' := e_i - e_a$ to obtain a new basis and dual basis in which $\fsub{m} \fatsemi e_a' \sim \fsub{m} \fatsemi e_a \not\sim 0$ and $\fsub{m} \fatsemi e_i' \sim \fsub{m} \fatsemi (e_i - e_a) \sim -\fsub{m} \fatsemi e_a \not\sim 0$.

We can now use this basis to resolve the identity on $A$ to give $\fsub{m} \sim \sum_i (\fsub{m} \fatsemi e_i) \cdot \fsub{\rho_i}$ and $\fsub{h} \sim \sum_i (\fsub{m} \fatsemi e_i) \cdot \fsub{\rho_i} \otimes (\fsub{m} \fatsemi e_i)^{-1} \cdot (\fsub{h} \fatsemi (e_i \otimes \id_B))$. Let $g_i := (\fsub{m} \fatsemi e_i)^{-1} \cdot (\fsub{h} \fatsemi (e_i \otimes \id_B))$.

For any $\pi \in c_\mathbf{B}^*$:
\begin{equation}
\begin{split}
g_i \fatsemi \fsub{\pi} &\sim \left(\fsub{m} \fatsemi e_i\right)^{-1} \cdot \left(\fsub{h} \fatsemi \left(e_i \otimes \id_B\right)\right) \fatsemi \fsub{\pi} \\
&\sim \left(\fsub{m} \fatsemi e_i\right)^{-1} \cdot \left(\fsub{h} \fatsemi \left(e_i \otimes \fsub{\pi}\right)\right) \\
&\sim \left(\fsub{m} \fatsemi e_i\right)^{-1} \cdot \left(\fsub{m} \fatsemi e_i\right) \\
&\sim 1
\end{split}
\end{equation}

Since each $g_i$ may only exist in $\sub(\catc)$, this may not be enough to say $g_i \in c_\mathbf{B}$. However, by Proposition \ref{prop:bintests} there is some $g_i' \in \catc(I,B)$ such that $\fsub{g_i'} \sim \fsub{\lambda_i \cdot \maxmix_B} + g_i$. Given $\mu \cdot \maxmix_B \in c_\mathbf{B}$ by flatness, we would have $\left(1 + \mu^{-1} \cdot \lambda_i\right)^{-1} \cdot g_i' \in c_\mathbf{B}$ (in the case where $1 + \mu^{-1} \cdot \lambda_i = 0$ has no inverse, we could have chosen a different value for $\lambda_i$ such as $\lambda_i + 1$).

\begin{equation}
\begin{split}
\fsub{h} \sim& \sum_i \left(\fsub{m} \fatsemi e_i\right) \cdot \fsub{\rho_i} \otimes \left(\fsub{m} \fatsemi e_i\right)^{-1} \cdot \left(\fsub{h} \fatsemi \left(e_i \otimes \id_B\right)\right) \\
\sim& \sum_i (\fsub{m} \fatsemi e_i) \cdot \fsub{\rho_i} \otimes g_i \\
\sim& \sum_i (\fsub{m} \fatsemi e_i) \cdot \fsub{\rho_i} \otimes \fsub{g_i'} - \sum_i (\fsub{m} \fatsemi e_i) \cdot \fsub{\rho_i} \otimes \fsub{\lambda_i \cdot \maxmix_B} \\
\sim& \sum_i (\fsub{m} \fatsemi e_i) \cdot \fsub{(1 + \mu^{-1} \cdot \lambda_i) \cdot \rho_i} \otimes \fsub{(1 + \mu^{-1} \cdot \lambda_i)^{-1} \cdot g_i'} \\ &- \sum_i (\fsub{m} \fatsemi e_i) \cdot \fsub{\mu^{-1} \cdot \lambda_i \cdot \rho_i} \otimes \fsub{\mu \cdot \maxmix_B}
\end{split}
\end{equation}



This now gives $h$ in the form of the definition for $c_\mathbf{A} \obslash c_\mathbf{B}$. For the final condition of that definition, we take the expansion of $m$ under the dual basis on $A$:
\begin{equation}
\begin{split}
\fsub{m} &\sim \sum_i \left(\fsub{m} \fatsemi e_i\right) \cdot \fsub{\rho_i} \\
&\sim \sum_i \left(\fsub{m} \fatsemi e_i\right) \cdot \fsub{(1 + \mu^{-1} \cdot \lambda_i) \cdot \rho_i} - \sum_i \left(\fsub{m} \fatsemi e_i\right) \cdot \fsub{\mu^{-1} \cdot \lambda_i \cdot \rho_i}
\end{split}
\end{equation}

\end{proof}

\begin{lemma}\label{lemma:one_way_normalises_semi_local}
$c_\mathbf{A} < c_\mathbf{B} \subseteq \left(c_\mathbf{A}^* \triangleleft c_\mathbf{B}^*\right)^*$.
\end{lemma}

\begin{proof}
Consider any $k \in c_\mathbf{A} < c_\mathbf{B}$ with residual $r \in c_\mathbf{A}$ and any $h \in c_\mathbf{A}^* \triangleleft c_\mathbf{B}^*$ which decomposes into $m \in c_{\mathbf{A^* \parr Z}}$ and $n \in c_{\mathbf{Z^* \parr B^*}}$. For any $\rho \in c_\mathbf{Z}$ (in particular, the elements of the canonical basis from \ref{apc:basis}), the following state is in $c_\mathbf{B}^*$:
\begin{equation}
\begin{tikzpicture}
	\begin{pgfonlayer}{nodelayer}
		\node [style=none] (0) at (-1, 0) {};
		\node [style=none] (1) at (0, 0) {};
		\node [style=none] (2) at (0, -1) {};
		\node [style=none] (3) at (-1, -1) {};
		\node [style=none] (4) at (-0.5, -0.5) {$\rho$};
		\node [style=none] (5) at (1, 0) {};
		\node [style=none] (6) at (3, 0) {};
		\node [style=none] (7) at (3, -1) {};
		\node [style=none] (8) at (1, -1) {};
		\node [style=none] (9) at (2, -0.5) {$n$};
		\node [style=none] (10) at (-0.5, 0) {};
		\node [style=none] (11) at (1.5, 0) {};
		\node [style=none] (12) at (2.5, 0) {};
		\node [style=none] (13) at (2.5, 1) {};
		\node [style=none] (14) at (0.5, 1) {};
		\node [style=none] (15) at (2.5, 0.5) {};
		\node [style=none] (16) at (-0.1, 0.4) {$Z$};
		\node [style=none] (17) at (2, 0.4) {$Z^*$};
		\node [style=none] (18) at (3.05, 0.4) {$B^*$};
	\end{pgfonlayer}
	\begin{pgfonlayer}{edgelayer}
		\draw [style=empty] (1.center)
			 to (2.center)
			 to (3.center)
			 to (0.center)
			 to cycle;
		\draw [style=empty] (6.center)
			 to (7.center)
			 to (8.center)
			 to (5.center)
			 to cycle;
		\draw [style=arr, bend left=45] (10.center) to (14.center);
		\draw [bend left=45] (14.center) to (11.center);
		\draw (15.center) to (12.center);
		\draw [style=arr] (13.center) to (15.center);
	\end{pgfonlayer}
\end{tikzpicture}
\end{equation}

Composing this with $k$ will yield $r$ by one-way signalling. Since this holds for every state in the basis of $Z$ and each basis element is in $\left\{\discard_Z\right\}^*$, we can infer that:
\begin{equation}
\begin{tikzpicture}
	\begin{pgfonlayer}{nodelayer}
		\node [style=none] (0) at (-2, 0) {};
		\node [style=none] (1) at (0, 0) {};
		\node [style=none] (2) at (0, -1) {};
		\node [style=none] (3) at (-2, -1) {};
		\node [style=none] (4) at (-1, -0.5) {$k$};
		\node [style=none] (5) at (-1.5, 0) {};
		\node [style=none] (6) at (-0.5, 0.25) {};
		\node [style=none] (7) at (1, 0) {};
		\node [style=none] (8) at (3, 0) {};
		\node [style=none] (9) at (3, -1) {};
		\node [style=none] (10) at (1, -1) {};
		\node [style=none] (11) at (2, -0.5) {$n$};
		\node [style=none] (12) at (1.5, 0) {};
		\node [style=none] (13) at (2.5, 0.25) {};
		\node [style=none] (14) at (-1.5, 0.5) {};
		\node [style=none] (15) at (-1.5, 1.5) {};
		\node [style=none] (16) at (1, 1.25) {};
		\node [style=none] (17) at (1.5, 1.5) {};
		\node [style=none] (18) at (1.5, 0.5) {};
		\node [style=none] (19) at (-1.15, 0.4) {$A$};
		\node [style=none] (20) at (-0.15, 0.4) {$B$};
		\node [style=none] (21) at (2, 0.4) {$Z^*$};
		\node [style=none] (22) at (3, 0.4) {$B^*$};
		\node [style=none] (23) at (-0.5, 0) {};
		\node [style=none] (24) at (2.5, 0) {};
	\end{pgfonlayer}
	\begin{pgfonlayer}{edgelayer}
		\draw [style=empty] (0.center) to (1.center);
		\draw [style=empty] (1.center) to (2.center);
		\draw [style=empty] (2.center) to (3.center);
		\draw [style=empty] (3.center) to (0.center);
		\draw [style=empty] (7.center) to (8.center);
		\draw [style=empty] (8.center) to (9.center);
		\draw [style=empty] (9.center) to (10.center);
		\draw [style=empty] (10.center) to (7.center);
		\draw (14.center) to (15.center);
		\draw [in=90, out=0] (16.center) to (13.center);
		\draw [style=arr, in=180, out=90] (6.center) to (16.center);
		\draw [style=arr] (5.center) to (14.center);
		\draw [style=arr] (17.center) to (18.center);
		\draw (18.center) to (12.center);
		\draw (6.center) to (23.center);
		\draw (13.center) to (24.center);
	\end{pgfonlayer}
\end{tikzpicture} = \begin{tikzpicture}
	\begin{pgfonlayer}{nodelayer}
		\node [style=none] (0) at (-1, 0) {};
		\node [style=none] (1) at (0, 0) {};
		\node [style=none] (2) at (0, -1) {};
		\node [style=none] (3) at (-1, -1) {};
		\node [style=none] (4) at (-0.5, -0.5) {$r$};
		\node [style=none, label={[label distance=-3pt]45:$A$}] (5) at (-0.5, 0) {};
		\node [style=none] (6) at (-0.5, 0.5) {};
		\node [style=none] (7) at (-0.5, 1) {};
		\node [style=maxmix] (8) at (1, -0.5) {};
		\node [style=none] (9) at (1, 0.5) {};
		\node [style=none] (10) at (1, 1) {};
    \node [style=none, label={[label distance=-3pt]45:$Z^*$}] (11) at (1, 0) {};
	\end{pgfonlayer}
	\begin{pgfonlayer}{edgelayer}
		\draw [style=empty] (3.center)
			 to (0.center)
			 to (1.center)
			 to (2.center)
			 to cycle;
		\draw [style=arr] (5.center) to (6.center);
		\draw [style=arr] (10.center) to (9.center);
		\draw (7.center) to (6.center);
		\draw (9.center) to (8.center);
	\end{pgfonlayer}
\end{tikzpicture}
\end{equation}

$r \otimes \maxmix_{Z^*} \in c_\mathbf{A} \otimes (c_\mathbf{Z}^*) \subseteq c_{\mathbf{A^* \parr Z}}^*$ and therefore it maps $m$ to $1$. In summary:
\begin{equation}
\begin{tikzpicture}
	\begin{pgfonlayer}{nodelayer}
		\node [style=none] (0) at (-2, 0) {};
		\node [style=none] (1) at (0, 0) {};
		\node [style=none] (2) at (0, -1) {};
		\node [style=none] (3) at (-2, -1) {};
		\node [style=none] (4) at (-1, -0.5) {$k$};
		\node [style=none] (5) at (-1.5, 0.25) {};
		\node [style=none] (6) at (-0.5, 0.25) {};
		\node [style=none] (7) at (1, 0) {};
		\node [style=none] (8) at (3, 0) {};
		\node [style=none] (9) at (3, -1) {};
		\node [style=none] (10) at (1, -1) {};
		\node [style=none] (11) at (2, -0.5) {$h$};
		\node [style=none] (12) at (1.5, 0.25) {};
		\node [style=none] (13) at (2.5, 0.25) {};
		\node [style=none] (16) at (1, 1.5) {};
		\node [style=none] (17) at (0, 1.5) {};
		\node [style=none] (18) at (-1.5, 0) {};
		\node [style=none] (19) at (-0.5, 0) {};
		\node [style=none] (20) at (1.5, 0) {};
		\node [style=none] (21) at (2.5, 0) {};
		\node [style=none] (22) at (-1.15, 0.4) {$A$};
		\node [style=none] (23) at (-0.15, 0.4) {$B$};
		\node [style=none] (24) at (2, 0.4) {$A^*$};
		\node [style=none] (25) at (3, 0.4) {$B^*$};
	\end{pgfonlayer}
	\begin{pgfonlayer}{edgelayer}
		\draw [style=empty] (0.center) to (1.center);
		\draw [style=empty] (1.center) to (2.center);
		\draw [style=empty] (2.center) to (3.center);
		\draw [style=empty] (3.center) to (0.center);
		\draw [style=empty] (7.center) to (8.center);
		\draw [style=empty] (8.center) to (9.center);
		\draw [style=empty] (9.center) to (10.center);
		\draw [style=empty] (10.center) to (7.center);
		\draw [in=90, out=0] (16.center) to (13.center);
		\draw [style=arr, in=180, out=90] (6.center) to (16.center);
		\draw [style=arr, in=180, out=90] (5.center) to (17.center);
		\draw [in=90, out=0] (17.center) to (12.center);
		\draw (5.center) to (18.center);
		\draw (6.center) to (19.center);
		\draw (12.center) to (20.center);
		\draw (13.center) to (21.center);
	\end{pgfonlayer}
\end{tikzpicture} = \begin{tikzpicture}
	\begin{pgfonlayer}{nodelayer}
		\node [style=none] (0) at (-2, 0) {};
		\node [style=none] (1) at (0, 0) {};
		\node [style=none] (2) at (0, -1) {};
		\node [style=none] (3) at (-2, -1) {};
		\node [style=none] (4) at (-1, -0.5) {$k$};
		\node [style=none] (5) at (-1.5, 0.25) {};
		\node [style=none] (6) at (-0.5, 0.25) {};
		\node [style=none] (7) at (4, 0) {};
		\node [style=none] (8) at (6, 0) {};
		\node [style=none] (9) at (6, -1) {};
		\node [style=none] (10) at (4, -1) {};
		\node [style=none] (11) at (5, -0.5) {$n$};
		\node [style=none] (12) at (4.5, 0.25) {};
		\node [style=none] (13) at (5.5, 0.25) {};
		\node [style=none] (14) at (0, 1.5) {};
		\node [style=none] (16) at (2.5, 1.5) {};
		\node [style=none] (18) at (3.5, 1.25) {};
		\node [style=none] (19) at (1, 0) {};
		\node [style=none] (20) at (3, 0) {};
		\node [style=none] (21) at (3, -1) {};
		\node [style=none] (22) at (1, -1) {};
		\node [style=none] (23) at (2, -0.5) {$m$};
		\node [style=none] (24) at (1.5, 0.25) {};
		\node [style=none] (25) at (2.5, 0.25) {};
		\node [style=none] (26) at (-1.5, 0) {};
		\node [style=none] (27) at (-0.5, 0) {};
		\node [style=none] (28) at (1.5, 0) {};
		\node [style=none] (29) at (2.5, 0) {};
		\node [style=none] (30) at (4.5, 0) {};
		\node [style=none] (31) at (5.5, 0) {};
		\node [style=none] (32) at (-1.15, 0.4) {$A$};
		\node [style=none] (33) at (-0.15, 0.4) {$B$};
		\node [style=none] (34) at (2, 0.4) {$A^*$};
		\node [style=none] (35) at (2.85, 0.4) {$Z$};
		\node [style=none] (36) at (5, 0.4) {$Z^*$};
		\node [style=none] (37) at (6, 0.4) {$B^*$};
	\end{pgfonlayer}
	\begin{pgfonlayer}{edgelayer}
		\draw [style=empty] (0.center) to (1.center);
		\draw [style=empty] (1.center) to (2.center);
		\draw [style=empty] (2.center) to (3.center);
		\draw [style=empty] (3.center) to (0.center);
		\draw [style=empty] (7.center) to (8.center);
		\draw [style=empty] (8.center) to (9.center);
		\draw [style=empty] (9.center) to (10.center);
		\draw [style=empty] (10.center) to (7.center);
		\draw [in=90, out=0, looseness=0.75] (16.center) to (13.center);
		\draw [bend left=45] (18.center) to (12.center);
		\draw [style=arr, in=180, out=90, looseness=0.75] (6.center) to (16.center);
		\draw [style=arr, in=-180, out=90] (5.center) to (14.center);
		\draw [style=empty] (19.center) to (20.center);
		\draw [style=empty] (20.center) to (21.center);
		\draw [style=empty] (21.center) to (22.center);
		\draw [style=empty] (22.center) to (19.center);
		\draw [in=90, out=0] (14.center) to (24.center);
		\draw [style=arr, bend left=45] (25.center) to (18.center);
		\draw (5.center) to (26.center);
		\draw (6.center) to (27.center);
		\draw (24.center) to (28.center);
		\draw (25.center) to (29.center);
		\draw (12.center) to (30.center);
		\draw (13.center) to (31.center);
	\end{pgfonlayer}
\end{tikzpicture} = \begin{tikzpicture}
	\begin{pgfonlayer}{nodelayer}
		\node [style=none] (0) at (-1, 0) {};
		\node [style=none] (1) at (0, 0) {};
		\node [style=none] (2) at (0, -1) {};
		\node [style=none] (3) at (-1, -1) {};
		\node [style=none] (4) at (-0.5, -0.5) {$r$};
		\node [style=none] (5) at (-0.5, 0.25) {};
		\node [style=none] (6) at (0.5, 1.25) {};
		\node [style=maxmix] (8) at (4, -0.5) {};
		\node [style=none] (10) at (3.25, 1.25) {};
		\node [style=none] (11) at (4, 0.25) {};
		\node [style=none] (12) at (1, 0) {};
		\node [style=none] (13) at (3, 0) {};
		\node [style=none] (14) at (3, -1) {};
		\node [style=none] (15) at (1, -1) {};
		\node [style=none] (16) at (2, -0.5) {$m$};
		\node [style=none] (17) at (1.5, 0.25) {};
		\node [style=none] (20) at (2.5, 0.25) {};
		\node [style=none] (21) at (-0.5, 0) {};
		\node [style=none] (22) at (1.5, 0) {};
		\node [style=none] (23) at (2.5, 0) {};
		\node [style=none] (24) at (-0.15, 0.4) {$A$};
		\node [style=none] (25) at (2, 0.4) {$A^*$};
		\node [style=none] (26) at (2.85, 0.4) {$Z$};
		\node [style=none] (27) at (4.5, 0.4) {$Z^*$};
	\end{pgfonlayer}
	\begin{pgfonlayer}{edgelayer}
		\draw [style=empty] (3.center)
			 to (0.center)
			 to (1.center)
			 to (2.center)
			 to cycle;
		\draw [style=arr, bend left=45] (5.center) to (6.center);
		\draw [style=empty] (15.center)
			 to (12.center)
			 to (13.center)
			 to (14.center)
			 to cycle;
		\draw [bend left=45] (6.center) to (17.center);
		\draw (11.center) to (8);
		\draw [in=90, out=0] (10.center) to (11.center);
		\draw [style=arr, in=-180, out=90] (20.center) to (10.center);
		\draw (5.center) to (21.center);
		\draw (17.center) to (22.center);
		\draw (20.center) to (23.center);
	\end{pgfonlayer}
\end{tikzpicture} = 1
\end{equation}
\end{proof}

\begin{lemma}\label{lemma:normalisers_of_semi_local_are_one_way}
$\left(c_\mathbf{A} \triangleleft c_\mathbf{B}\right)^* \subseteq c_\mathbf{A}^* < c_\mathbf{B}^*$.
\end{lemma}

\begin{proof}
Consider any $h \in \left(c_\mathbf{A} \triangleleft c_\mathbf{B}\right)^*$.

Let $\mathbf{2} := \mathbf{I \oplus I}$ be the coproduct of monoidal units, which is a first-order type since coproducts preserves the first-order property (Proposition \ref{prop:coproduct_first_order}). For any two $b, b' \in c_\mathbf{B}$, their coproduct is $\langle b, b' \rangle \in c_{\mathbf{2^* \parr B}}$. For any $a \in c_{\mathbf{A \parr 2}}$:
\begin{equation}
\begin{tikzpicture}
	\begin{pgfonlayer}{nodelayer}
		\node [style=none] (0) at (-2, 0) {};
		\node [style=none] (1) at (0, 0) {};
		\node [style=none] (2) at (0, -1) {};
		\node [style=none] (3) at (-2, -1) {};
		\node [style=none] (4) at (-1, -0.5) {$a$};
		\node [style=none, label={[label distance=-3pt]45:$A$}] (5) at (-1.5, 0) {};
		\node [style=none, label={[label distance=-3pt]45:$2$}] (6) at (-0.5, 0) {};
		\node [style=none] (7) at (1, 0) {};
		\node [style=none] (8) at (3, 0) {};
		\node [style=none] (9) at (3, -1) {};
		\node [style=none] (10) at (1, -1) {};
		\node [style=none] (11) at (2, -0.5) {$\langle b,b' \rangle$};
		\node [style=none, label={[label distance=-3pt]45:$2^*$}] (12) at (1.5, 0) {};
		\node [style=none, label={[label distance=-3pt]45:$B$}] (13) at (2.5, 0) {};
		\node [style=none] (14) at (-1.5, 0.5) {};
		\node [style=none] (15) at (-1.5, 1) {};
		\node [style=none] (16) at (0.5, 1) {};
		\node [style=none] (17) at (2.5, 0.5) {};
		\node [style=none] (18) at (2.5, 1) {};
	\end{pgfonlayer}
	\begin{pgfonlayer}{edgelayer}
		\draw [style=empty] (2.center)
			 to (3.center)
			 to (0.center)
			 to (1.center)
			 to cycle;
		\draw [style=empty] (9.center)
			 to (10.center)
			 to (7.center)
			 to (8.center)
			 to cycle;
		\draw [style=arr] (5.center) to (14.center);
		\draw [style=arr, bend left=45] (6.center) to (16.center);
		\draw [style=arr] (13.center) to (17.center);
		\draw (14.center) to (15.center);
		\draw [bend left=45] (16.center) to (12.center);
		\draw (17.center) to (18.center);
	\end{pgfonlayer}
\end{tikzpicture} \in c_{\mathbf{A} \triangleleft \mathbf{B}}
\end{equation}

Taking the inner product of this morphism with $h$ will give $1$. Since this holds for any $a \in c_{\mathbf{A \parr 2}}$:
\begin{equation}
\begin{tikzpicture}
	\begin{pgfonlayer}{nodelayer}
		\node [style=none] (0) at (-2, 0) {};
		\node [style=none] (1) at (0, 0) {};
		\node [style=none] (2) at (0, -1) {};
		\node [style=none] (3) at (-2, -1) {};
		\node [style=none] (4) at (-1, -0.5) {$h$};
		\node [style=none] (5) at (-1.5, 0) {};
		\node [style=none] (6) at (-0.5, 0.25) {};
		\node [style=none] (7) at (1, 0) {};
		\node [style=none] (8) at (3, 0) {};
		\node [style=none] (9) at (3, -1) {};
		\node [style=none] (10) at (1, -1) {};
		\node [style=none] (11) at (2, -0.5) {$\langle b,b' \rangle$};
		\node [style=none] (12) at (1.5, 0) {};
		\node [style=none] (13) at (2.5, 0.25) {};
		\node [style=none] (14) at (-1.5, 0.5) {};
		\node [style=none] (15) at (-1.5, 1.5) {};
		\node [style=none] (16) at (1, 1.25) {};
		\node [style=none] (17) at (1.5, 1.5) {};
		\node [style=none] (18) at (1.5, 0.5) {};
		\node [style=none] (19) at (-0.5, 0) {};
		\node [style=none] (20) at (2.5, 0) {};
		\node [style=none] (21) at (-1, 0.4) {$A^*$};
		\node [style=none] (22) at (0, 0.4) {$B^*$};
		\node [style=none] (23) at (2, 0.4) {$2^*$};
		\node [style=none] (24) at (2.85, 0.4) {$B$};
	\end{pgfonlayer}
	\begin{pgfonlayer}{edgelayer}
		\draw [style=empty] (0.center) to (1.center);
		\draw [style=empty] (1.center) to (2.center);
		\draw [style=empty] (2.center) to (3.center);
		\draw [style=empty] (3.center) to (0.center);
		\draw [style=empty] (7.center) to (8.center);
		\draw [style=empty] (8.center) to (9.center);
		\draw [style=empty] (9.center) to (10.center);
		\draw [style=empty] (10.center) to (7.center);
		\draw [style=arr] (15.center) to (14.center);
		\draw [style=arr, in=360, out=90] (13.center) to (16.center);
		\draw [in=90, out=-180] (16.center) to (6.center);
		\draw (14.center) to (5.center);
		\draw (18.center) to (12.center);
		\draw [style=arr] (17.center) to (18.center);
		\draw (6.center) to (19.center);
		\draw (13.center) to (20.center);
	\end{pgfonlayer}
\end{tikzpicture} \in c_{\mathbf{A \parr 2}}^* = c_{\mathbf{A^* \otimes 2^*}}
\end{equation}

$\mathbf{2}$ is first-order, so $c_\mathbf{2}^* = \left\{\discard_2\right\}$. By Theorem \ref{thrm:no_interaction_with_trivial}, the above morphism must separate into $\pi_{h,b,b'} \otimes \maxmix_2$ for some $\pi_{h,b,b'} \in c_\mathbf{A}^*$.

\begin{equation}
\begin{split}
\begin{tikzpicture}
	\begin{pgfonlayer}{nodelayer}
		\node [style=none] (0) at (-2, 0) {};
		\node [style=none] (1) at (0, 0) {};
		\node [style=none] (2) at (0, -1) {};
		\node [style=none] (3) at (-2, -1) {};
		\node [style=none] (4) at (-1, -0.5) {$h$};
		\node [style=none, label={[label distance=-3pt]45:$A^*$}] (5) at (-1.5, 0) {};
		\node [style=none, label={[label distance=-3pt]45:$B^*$}] (6) at (-0.5, 0) {};
		\node [style=none] (7) at (1, 0) {};
		\node [style=none] (8) at (2, 0) {};
		\node [style=none] (9) at (2, -1) {};
		\node [style=none] (10) at (1, -1) {};
		\node [style=none] (11) at (1.5, -0.5) {$b$};
		\node [style=none, label={[label distance=-3pt]45:$B$}] (13) at (1.5, 0) {};
		\node [style=none] (14) at (-1.5, 0.5) {};
		\node [style=none] (15) at (-1.5, 1) {};
		\node [style=none] (16) at (0.5, 1) {};
	\end{pgfonlayer}
	\begin{pgfonlayer}{edgelayer}
		\draw [style=empty] (0.center) to (1.center);
		\draw [style=empty] (1.center) to (2.center);
		\draw [style=empty] (2.center) to (3.center);
		\draw [style=empty] (3.center) to (0.center);
		\draw [style=empty] (7.center) to (8.center);
		\draw [style=empty] (8.center) to (9.center);
		\draw [style=empty] (9.center) to (10.center);
		\draw [style=empty] (10.center) to (7.center);
		\draw [style=arr] (15.center) to (14.center);
		\draw [style=arr, bend right=45] (13.center) to (16.center);
		\draw [bend right=45] (16.center) to (6.center);
		\draw (14.center) to (5.center);
	\end{pgfonlayer}
\end{tikzpicture} &= \begin{tikzpicture}
	\begin{pgfonlayer}{nodelayer}
		\node [style=none] (0) at (-2, 0) {};
		\node [style=none] (1) at (0, 0) {};
		\node [style=none] (2) at (0, -1) {};
		\node [style=none] (3) at (-2, -1) {};
		\node [style=none] (4) at (-1, -0.5) {$h$};
		\node [style=none, label={[label distance=-3pt]45:$A^*$}] (5) at (-1.5, 0) {};
		\node [style=none, label={[label distance=-3pt]45:$B^*$}] (6) at (-0.5, 0) {};
		\node [style=none] (7) at (3, 0) {};
		\node [style=none] (8) at (5, 0) {};
		\node [style=none] (9) at (5, -1) {};
		\node [style=none] (10) at (3, -1) {};
		\node [style=none] (11) at (4, -0.5) {$\langle b,b' \rangle$};
		\node [style=none, label={[label distance=-3pt]45:$2^*$}] (12) at (3.5, 0) {};
		\node [style=none, label={[label distance=-3pt]45:$B$}] (13) at (4.5, 0) {};
		\node [style=none] (14) at (-1.5, 0.5) {};
		\node [style=none] (15) at (-1.5, 1) {};
		\node [style=none] (16) at (2, 1.5) {};
		\node [style=none] (17) at (1, 0) {};
		\node [style=none] (18) at (2, 0) {};
		\node [style=none] (19) at (2, -1) {};
		\node [style=none] (20) at (1, -1) {};
		\node [style=none] (21) at (1.5, -0.5) {$\iota_1$};
		\node [style=none, label={[label distance=-3pt]45:$2$}] (22) at (1.5, 0) {};
		\node [style=none] (23) at (2.5, 1) {};
	\end{pgfonlayer}
	\begin{pgfonlayer}{edgelayer}
		\draw [style=empty] (0.center) to (1.center);
		\draw [style=empty] (1.center) to (2.center);
		\draw [style=empty] (2.center) to (3.center);
		\draw [style=empty] (3.center) to (0.center);
		\draw [style=empty] (7.center) to (8.center);
		\draw [style=empty] (8.center) to (9.center);
		\draw [style=empty] (9.center) to (10.center);
		\draw [style=empty] (10.center) to (7.center);
		\draw [style=arr] (15.center) to (14.center);
		\draw [style=arr, in=360, out=90] (13.center) to (16.center);
		\draw [in=90, out=-180] (16.center) to (6.center);
		\draw (14.center) to (5.center);
		\draw [style=empty] (17.center)
			 to (18.center)
			 to (19.center)
			 to (20.center)
			 to cycle;
		\draw [style=arr, bend left=45] (22.center) to (23.center);
		\draw [bend left=45] (23.center) to (12.center);
	\end{pgfonlayer}
\end{tikzpicture} \\
&= \begin{tikzpicture}
	\begin{pgfonlayer}{nodelayer}
		\node [style=none] (0) at (-2, 0) {};
		\node [style=none] (1) at (0, 0) {};
		\node [style=none] (2) at (0, -1) {};
		\node [style=none] (3) at (-2, -1) {};
		\node [style=none] (4) at (-1, -0.5) {$\pi_{h,b,b'}$};
		\node [style=none, label={[label distance=-3pt]45:$A^*$}] (5) at (-1, 0) {};
		\node [style=none] (6) at (-1, 0.5) {};
		\node [style=none] (7) at (-1, 1) {};
		\node [style=maxmix] (8) at (3, -0.5) {};
		\node [style=none, label={[label distance=-3pt]45:$2^*$}] (11) at (3, 0) {};
		\node [style=none] (12) at (1, 0) {};
		\node [style=none] (13) at (2, 0) {};
		\node [style=none] (14) at (2, -1) {};
		\node [style=none] (15) at (1, -1) {};
		\node [style=none] (16) at (1.5, -0.5) {$\iota_1$};
		\node [style=none, label={[label distance=-3pt]45:$2$}] (17) at (1.5, 0) {};
		\node [style=none] (18) at (2.25, 1) {};
	\end{pgfonlayer}
	\begin{pgfonlayer}{edgelayer}
		\draw [style=empty] (3.center)
			 to (0.center)
			 to (1.center)
			 to (2.center)
			 to cycle;
		\draw [style=arr] (5.center) to (6.center);
		\draw (7.center) to (6.center);
		\draw [style=empty] (15.center)
			 to (12.center)
			 to (13.center)
			 to (14.center)
			 to cycle;
		\draw [style=arr, in=-180, out=90] (17.center) to (18.center);
		\draw [in=90, out=0] (18.center) to (11.center);
		\draw (11.center) to (8);
	\end{pgfonlayer}
\end{tikzpicture} \\
&= \begin{tikzpicture}
	\begin{pgfonlayer}{nodelayer}
		\node [style=none] (0) at (-2, 0) {};
		\node [style=none] (1) at (0, 0) {};
		\node [style=none] (2) at (0, -1) {};
		\node [style=none] (3) at (-2, -1) {};
		\node [style=none] (4) at (-1, -0.5) {$\pi_{h,b,b'}$};
		\node [style=none, label={[label distance=-3pt]45:$A^*$}] (5) at (-1, 0) {};
		\node [style=none] (6) at (-1, 0.5) {};
		\node [style=none] (7) at (-1, 1) {};
		\node [style=maxmix] (8) at (3, -0.5) {};
		\node [style=none, label={[label distance=-3pt]45:$2^*$}] (11) at (3, 0) {};
		\node [style=none] (12) at (1, 0) {};
		\node [style=none] (13) at (2, 0) {};
		\node [style=none] (14) at (2, -1) {};
		\node [style=none] (15) at (1, -1) {};
		\node [style=none] (16) at (1.5, -0.5) {$\iota_2$};
		\node [style=none, label={[label distance=-3pt]45:$2$}] (17) at (1.5, 0) {};
		\node [style=none] (18) at (2.25, 1) {};
	\end{pgfonlayer}
	\begin{pgfonlayer}{edgelayer}
		\draw [style=empty] (3.center)
			 to (0.center)
			 to (1.center)
			 to (2.center)
			 to cycle;
		\draw [style=arr] (5.center) to (6.center);
		\draw (7.center) to (6.center);
		\draw [style=empty] (15.center)
			 to (12.center)
			 to (13.center)
			 to (14.center)
			 to cycle;
		\draw [style=arr, in=-180, out=90] (17.center) to (18.center);
		\draw [in=90, out=0] (18.center) to (11.center);
		\draw (11.center) to (8);
	\end{pgfonlayer}
\end{tikzpicture} \\
&= \begin{tikzpicture}
	\begin{pgfonlayer}{nodelayer}
		\node [style=none] (0) at (-2, 0) {};
		\node [style=none] (1) at (0, 0) {};
		\node [style=none] (2) at (0, -1) {};
		\node [style=none] (3) at (-2, -1) {};
		\node [style=none] (4) at (-1, -0.5) {$h$};
		\node [style=none, label={[label distance=-3pt]45:$A^*$}] (5) at (-1.5, 0) {};
		\node [style=none, label={[label distance=-3pt]45:$B^*$}] (6) at (-0.5, 0) {};
		\node [style=none] (7) at (3, 0) {};
		\node [style=none] (8) at (5, 0) {};
		\node [style=none] (9) at (5, -1) {};
		\node [style=none] (10) at (3, -1) {};
		\node [style=none] (11) at (4, -0.5) {$\langle b,b' \rangle$};
		\node [style=none, label={[label distance=-3pt]45:$2^*$}] (12) at (3.5, 0) {};
		\node [style=none, label={[label distance=-3pt]45:$B$}] (13) at (4.5, 0) {};
		\node [style=none] (14) at (-1.5, 0.5) {};
		\node [style=none] (15) at (-1.5, 1) {};
		\node [style=none] (16) at (2, 1.5) {};
		\node [style=none] (17) at (1, 0) {};
		\node [style=none] (18) at (2, 0) {};
		\node [style=none] (19) at (2, -1) {};
		\node [style=none] (20) at (1, -1) {};
		\node [style=none] (21) at (1.5, -0.5) {$\iota_2$};
		\node [style=none, label={[label distance=-3pt]45:$2$}] (22) at (1.5, 0) {};
		\node [style=none] (23) at (2.5, 1) {};
	\end{pgfonlayer}
	\begin{pgfonlayer}{edgelayer}
		\draw [style=empty] (0.center) to (1.center);
		\draw [style=empty] (1.center) to (2.center);
		\draw [style=empty] (2.center) to (3.center);
		\draw [style=empty] (3.center) to (0.center);
		\draw [style=empty] (7.center) to (8.center);
		\draw [style=empty] (8.center) to (9.center);
		\draw [style=empty] (9.center) to (10.center);
		\draw [style=empty] (10.center) to (7.center);
		\draw [style=arr] (15.center) to (14.center);
		\draw [style=arr, in=360, out=90] (13.center) to (16.center);
		\draw [in=90, out=-180] (16.center) to (6.center);
		\draw (14.center) to (5.center);
		\draw [style=empty] (17.center)
			 to (18.center)
			 to (19.center)
			 to (20.center)
			 to cycle;
		\draw [style=arr, bend left=45] (22.center) to (23.center);
		\draw [bend left=45] (23.center) to (12.center);
	\end{pgfonlayer}
\end{tikzpicture} \\
&= \begin{tikzpicture}
	\begin{pgfonlayer}{nodelayer}
		\node [style=none] (0) at (-2, 0) {};
		\node [style=none] (1) at (0, 0) {};
		\node [style=none] (2) at (0, -1) {};
		\node [style=none] (3) at (-2, -1) {};
		\node [style=none] (4) at (-1, -0.5) {$h$};
		\node [style=none, label={[label distance=-3pt]45:$A^*$}] (5) at (-1.5, 0) {};
		\node [style=none, label={[label distance=-3pt]45:$B^*$}] (6) at (-0.5, 0) {};
		\node [style=none] (7) at (1, 0) {};
		\node [style=none] (8) at (2, 0) {};
		\node [style=none] (9) at (2, -1) {};
		\node [style=none] (10) at (1, -1) {};
		\node [style=none] (11) at (1.5, -0.5) {$b'$};
		\node [style=none, label={[label distance=-3pt]45:$B$}] (13) at (1.5, 0) {};
		\node [style=none] (14) at (-1.5, 0.5) {};
		\node [style=none] (15) at (-1.5, 1) {};
		\node [style=none] (16) at (0.5, 1) {};
	\end{pgfonlayer}
	\begin{pgfonlayer}{edgelayer}
		\draw [style=empty] (0.center) to (1.center);
		\draw [style=empty] (1.center) to (2.center);
		\draw [style=empty] (2.center) to (3.center);
		\draw [style=empty] (3.center) to (0.center);
		\draw [style=empty] (7.center) to (8.center);
		\draw [style=empty] (8.center) to (9.center);
		\draw [style=empty] (9.center) to (10.center);
		\draw [style=empty] (10.center) to (7.center);
		\draw [style=arr] (15.center) to (14.center);
		\draw [style=arr, bend right=45] (13.center) to (16.center);
		\draw [bend right=45] (16.center) to (6.center);
		\draw (14.center) to (5.center);
	\end{pgfonlayer}
\end{tikzpicture}
\end{split}
\end{equation}

Since this is the case for any choice of $b$ and $b'$, we can conclude that $h$ reduces to a constant left-residual $\pi_{h,-,-} \in c_\mathbf{A}^*$ and so $h \in c_\mathbf{A}^* < c_\mathbf{B}^*$.
\end{proof}

\begin{corollary}\label{corollary:one_way_closed}
$c_\mathbf{A} < c_\mathbf{B}$ is closed.
\end{corollary}

\begin{proof}
By Lemmas \ref{lemma:one_way_normalises_semi_local} and \ref{lemma:normalisers_of_semi_local_are_one_way}, $c_\mathbf{A} < c_\mathbf{B} = \left(c_\mathbf{A}^* \triangleleft c_\mathbf{B}^*\right)^*$. Since it is the dual of some other set, it must be closed.
\end{proof}

\begin{lemma}\label{lemma:semi_local_is_one_way}
$c_\mathbf{A} \triangleleft c_\mathbf{B} \subseteq c_\mathbf{A} < c_\mathbf{B}$.
\end{lemma}

\begin{proof}
Let $h \in c_\mathbf{A} \triangleleft c_\mathbf{B}$ decompose into $m \in c_{\mathbf{A \parr Z}}$ and $n \in c_{\mathbf{Z^* \parr B}}$. Applying any effect $\pi \in c_\mathbf{B}^*$ locally to $n$ must result in a causal morphism of type $\mathbf{Z^*}$. Since $\mathbf{Z}$ is first-order, this must be $\maxmix_{Z^*}$ regardless of $\pi$. Thus $h \fatsemi (\id_A \otimes \pi) = m \fatsemi (\id_A \otimes \discard_Z)$ gives a witness residual for $h \in c_\mathbf{A} < c_\mathbf{B}$.
\end{proof}

\begin{lemma}\label{lemma:one_way_normalises_pcs}
$c_\mathbf{A} < c_\mathbf{B} \subseteq \left(c_\mathbf{A}^* \obslash c_\mathbf{B}^*\right)^*$.
\end{lemma}

\begin{proof}
Take any $h \in c_\mathbf{A} < c_\mathbf{B}$ with residual $m \in c_\mathbf{A}$ and any $k \in c_\mathbf{A}^* \obslash c_\mathbf{B}^*$ decomposing as $k + \sum_{i \in \mathcal{I}^-} f_i^- \otimes g_i^- = \sum_{i \in \mathcal{I}^+} f_i^+ \otimes g_i^+$ with $n + \sum_{i \in \mathcal{I}^-} f_i^- = \sum_{i \in \mathcal{I}^+} f_i^+$, $n \in c_\mathbf{A}^*$, $\left\{g_i^\pm\right\}_{i \in \mathcal{I}^\pm} \subseteq c_\mathbf{B}^*$.

Take any $h \in c_\mathbf{A} < c_\mathbf{B}$ with residual $m \in c_\mathbf{A}$ and any $k \in c_\mathbf{A}^* \obslash c_\mathbf{B}^*$ decomposing as $\fsub{k} \sim \sum_{i \in \mathcal{I}} f_i \otimes \fsub{g_i}$ with $\{g_i\}_{i \in \mathcal{I}} \subseteq c_\mathbf{B}^*$ and $f \in c_\mathbf{A}^*$ such that $\fsub{f} \sim \sum_{i \in \mathcal{I}} f_i$.

\begin{equation}
\begin{split}
\fsub{k \fatsemi h^*} &\sim \left( \sum_{i \in \mathcal{I}} f_i \otimes \fsub{g_i} \right) \fatsemi \fsub{h^*} \\
&\sim \sum_{i \in \mathcal{I}} f_i \fatsemi \fsub{m^*} \\
&\sim \fsub{f \fatsemi m^*} \\
&\sim \fsub{1}
\end{split}
\end{equation}


So by faithfulness of the embedding into $\sub(\catc)$, $k \fatsemi h^* = 1$.
\end{proof}

\begin{theorem}[Restatement of Theorem \ref{thrm:one_way_self_dual_and_affine_semi_local}]\label{proof:thrm:one_way_self_dual_and_affine_semi_local}
$c_\mathbf{A} < c_\mathbf{B} = \left(c_\mathbf{A}^* < c_\mathbf{B}^*\right)^* = \left(c_\mathbf{A} \triangleleft c_\mathbf{B}\right)^{**} = c_\mathbf{A} \obslash c_\mathbf{B}$
\end{theorem}

\begin{proof}
Combining Lemmas \ref{lemma:one_way_normalises_semi_local} and \ref{lemma:normalisers_of_semi_local_are_one_way} gives $\left(c_\mathbf{A}^* < c_\mathbf{B}^*\right)^* = \left(c_\mathbf{A} \triangleleft c_\mathbf{B}\right)^{**}$. Then Corollary \ref{corollary:one_way_closed} and Lemma \ref{lemma:semi_local_is_one_way} give $\left(c_\mathbf{A} \triangleleft c_\mathbf{B}\right)^{**} \subseteq c_\mathbf{A} < c_\mathbf{B}$ with the other direction of inclusion given by Lemmas \ref{lemma:one_way_equals_pcs} and \ref{lemma:one_way_normalises_pcs}.
\end{proof}

Building up to the definition of a $\bv$-category from Blute, Panangaden, and Slavnov \cite{Blute2012}, we first note that $*$-autonomous categories are linearly distributive categories with negation \cite{Cockett1997}. The next step is to show that the one-way signalling operator forms a (non-symmetric) monoidal structure.

\begin{proposition}
$(\mathbf{<}, \mathbf{I}, \alpha, \lambda, \rho)$ forms a monoidal structure for $\caus[\catc]$.
\end{proposition}

\begin{proof}
We generate the bifunctor $\mathbf{<}$ by mapping objects as $\mathbf{A < B} := \left(A \otimes B, c_\mathbf{A} < c_\mathbf{B}\right)$ and copying the action on morphisms from $\otimes : \catc \times \catc \to \catc$. Given any $h \in c_{\mathbf{A < B}}$ with residual $m \in c_\mathbf{A}$, and morphisms $f : \mathbf{A} \to \mathbf{C}$ and $g : \mathbf{B} \to \mathbf{D}$, $h \fatsemi (f \otimes g)$ will be of type $\mathbf{C < D}$ with residual $m \fatsemi f \in c_\mathbf{C}$.

Naturality of $\alpha$, $\lambda$, and $\rho$ are inherited from $\catc$, as are the coherence identities. We just need to show that each of these is causal when using $\mathbf{<}$-types.


For $\alpha_{A, B, C} : \mathbf{(A < B) < C} \to \mathbf{A < (B < C)}$, take any $h \in c_{\mathbf{(A < B) < C}}$ with successive residuals $m \in c_{\mathbf{A < B}}$ and $n \in c_\mathbf{A}$. Any effect $k \in c_{\mathbf{B < C}}^* \simeq c_\mathbf{B}^* \obslash c_\mathbf{C}^*$ can be expressed as $\fsub{k} \sim \sum_i f_i \otimes \fsub{g_i}$ where $\{g_i\}_i \subseteq c_\mathbf{C}^*$ and $\fsub{f} \sim \sum_i f_i$ for some $f \in c_\mathbf{B}^*$. We want to show that connecting $h$ and $k$ (via the associator) yields a constant residual in $c_\mathbf{A}$ independent of $k$.

\begin{equation}
\begin{tikzpicture}
	\begin{pgfonlayer}{nodelayer}
		\node [style=none] (0) at (-3, 0) {};
		\node [style=none] (1) at (0, 0) {};
		\node [style=none] (2) at (0, -1) {};
		\node [style=none] (3) at (-3, -1) {};
		\node [style=none] (4) at (1, 0) {};
		\node [style=none] (5) at (3, 0) {};
		\node [style=none] (6) at (3, -1) {};
		\node [style=none] (7) at (1, -1) {};
		\node [style=none, label={[label distance=-3pt]45:$A$}] (8) at (-2.5, 0) {};
		\node [style=none, label={[label distance=-3pt]45:$B$}] (9) at (-1.5, 0) {};
		\node [style=none, label={[label distance=-3pt]45:$C$}] (10) at (-0.5, 0) {};
		\node [style=none, label={[label distance=-3pt]45:$B^*$}] (11) at (1.5, 0) {};
		\node [style=none, label={[label distance=-3pt]45:$C^*$}] (12) at (2.5, 0) {};
		\node [style=none] (13) at (-1.5, -0.5) {$\fsub{h}$};
		\node [style=none] (14) at (2, -0.5) {$\fsub{k}$};
		\node [style=none] (15) at (-2.5, 0.5) {};
		\node [style=none] (16) at (-2.5, 1) {};
		\node [style=none] (17) at (1, 1) {};
		\node [style=none] (18) at (0, 1) {};
	\end{pgfonlayer}
	\begin{pgfonlayer}{edgelayer}
		\draw [style=empty] (2.center)
			 to (3.center)
			 to (0.center)
			 to (1.center)
			 to cycle;
		\draw [style=empty] (6.center)
			 to (7.center)
			 to (4.center)
			 to (5.center)
			 to cycle;
		\draw [style=arr] (8.center) to (15.center);
		\draw [style=arr, in=180, out=90, looseness=0.75] (9.center) to (18.center);
		\draw [style=arr, in=-180, out=90, looseness=0.75] (10.center) to (17.center);
		\draw (15.center) to (16.center);
		\draw [in=90, out=0, looseness=0.75] (18.center) to (11.center);
		\draw [in=90, out=0, looseness=0.75] (17.center) to (12.center);
	\end{pgfonlayer}
\end{tikzpicture} \sim \sum_i \begin{tikzpicture}
	\begin{pgfonlayer}{nodelayer}
		\node [style=none] (0) at (-3, 0) {};
		\node [style=none] (1) at (0, 0) {};
		\node [style=none] (2) at (0, -1) {};
		\node [style=none] (3) at (-3, -1) {};
		\node [style=none] (4) at (1, 0) {};
		\node [style=none] (5) at (2, 0) {};
		\node [style=none] (6) at (2, -1) {};
		\node [style=none] (7) at (1, -1) {};
		\node [style=none, label={[label distance=-3pt]45:$A$}] (8) at (-2.5, 0) {};
		\node [style=none, label={[label distance=-3pt]45:$B$}] (9) at (-1.5, 0) {};
		\node [style=none, label={[label distance=-3pt]45:$C$}] (10) at (-0.5, 0) {};
		\node [style=none, label={[label distance=-3pt]45:$B^*$}] (11) at (1.5, 0) {};
		\node [style=none, label={[label distance=-3pt]45:$C^*$}] (12) at (3.5, 0) {};
		\node [style=none] (13) at (-1.5, -0.5) {$\fsub{h}$};
		\node [style=none] (14) at (1.5, -0.5) {$f_i$};
		\node [style=none] (15) at (-2.5, 0.5) {};
		\node [style=none] (16) at (-2.5, 1) {};
		\node [style=none] (17) at (1.5, 1) {};
		\node [style=none] (18) at (0, 1) {};
		\node [style=none] (19) at (3, 0) {};
		\node [style=none] (20) at (4, 0) {};
		\node [style=none] (21) at (4, -1) {};
		\node [style=none] (22) at (3, -1) {};
		\node [style=none] (23) at (3.5, -0.5) {$\fsub{g_i}$};
	\end{pgfonlayer}
	\begin{pgfonlayer}{edgelayer}
		\draw [style=empty] (2.center)
			 to (3.center)
			 to (0.center)
			 to (1.center)
			 to cycle;
		\draw [style=empty] (6.center)
			 to (7.center)
			 to (4.center)
			 to (5.center)
			 to cycle;
		\draw [style=arr] (8.center) to (15.center);
		\draw [style=arr, in=180, out=90, looseness=0.75] (9.center) to (18.center);
		\draw [style=arr, in=-180, out=90, looseness=0.75] (10.center) to (17.center);
		\draw (15.center) to (16.center);
		\draw [in=90, out=0, looseness=0.75] (18.center) to (11.center);
		\draw [in=90, out=0, looseness=0.75] (17.center) to (12.center);
		\draw [style=empty] (21.center)
			 to (22.center)
			 to (19.center)
			 to (20.center)
			 to cycle;
	\end{pgfonlayer}
\end{tikzpicture} \sim \sum_i \begin{tikzpicture}
	\begin{pgfonlayer}{nodelayer}
		\node [style=none] (0) at (-3, 0) {};
		\node [style=none] (1) at (-1, 0) {};
		\node [style=none] (2) at (-1, -1) {};
		\node [style=none] (3) at (-3, -1) {};
		\node [style=none] (4) at (0, 0) {};
		\node [style=none] (5) at (1, 0) {};
		\node [style=none] (6) at (1, -1) {};
		\node [style=none] (7) at (0, -1) {};
		\node [style=none, label={[label distance=-3pt]45:$A$}] (8) at (-2.5, 0) {};
		\node [style=none, label={[label distance=-3pt]45:$B$}] (9) at (-1.5, 0) {};
		\node [style=none, label={[label distance=-3pt]45:$B^*$}] (11) at (0.5, 0) {};
		\node [style=none] (13) at (-2, -0.5) {$\fsub{m}$};
		\node [style=none] (14) at (0.5, -0.5) {$f_i$};
		\node [style=none] (15) at (-2.5, 0.5) {};
		\node [style=none] (16) at (-2.5, 1) {};
		\node [style=none] (18) at (-0.5, 1) {};
	\end{pgfonlayer}
	\begin{pgfonlayer}{edgelayer}
		\draw [style=empty] (2.center)
			 to (3.center)
			 to (0.center)
			 to (1.center)
			 to cycle;
		\draw [style=empty] (6.center)
			 to (7.center)
			 to (4.center)
			 to (5.center)
			 to cycle;
		\draw [style=arr] (8.center) to (15.center);
		\draw [style=arr, bend left=45] (9.center) to (18.center);
		\draw (15.center) to (16.center);
		\draw [bend left=45] (18.center) to (11.center);
	\end{pgfonlayer}
\end{tikzpicture} \sim \begin{tikzpicture}
	\begin{pgfonlayer}{nodelayer}
		\node [style=none] (0) at (-1, 0) {};
		\node [style=none] (1) at (0, 0) {};
		\node [style=none] (2) at (0, -1) {};
		\node [style=none] (3) at (-1, -1) {};
		\node [style=none, label={[label distance=-3pt]45:$A$}] (8) at (-0.5, 0) {};
		\node [style=none] (13) at (-0.5, -0.5) {$\fsub{n}$};
		\node [style=none] (15) at (-0.5, 0.5) {};
		\node [style=none] (16) at (-0.5, 1) {};
	\end{pgfonlayer}
	\begin{pgfonlayer}{edgelayer}
		\draw [style=empty] (2.center)
			 to (3.center)
			 to (0.center)
			 to (1.center)
			 to cycle;
		\draw [style=arr] (8.center) to (15.center);
		\draw (15.center) to (16.center);
	\end{pgfonlayer}
\end{tikzpicture}
\end{equation}

For $\lambda_A : \mathbf{I < A} \to \mathbf{A}$, consider any $h \in c_{\mathbf{I < A}}$ and effect $\pi \in c_\mathbf{A}^*$. The residual of $h$ must be $1$ because it is the only causal state of $\mathbf{I}$, hence adjoining $h$ and $\pi$ (via $\lambda_A$) will yield $1$ for normalisation.

For $\rho_A : \mathbf{A < I} \to \mathbf{A}$, consider any $h \in c_{\mathbf{A < I}}$. The unique effect from $\mathbf{I}$ is $1$, so $h \fatsemi (\id_A \otimes 1) \fatsemi \rho_A = h \fatsemi \rho_A$ must be the residual, and hence $h \fatsemi \rho_A \in c_\mathbf{A}$.

The inverses $\alpha^{-1}$, $\lambda^{-1}$, and $\rho^{-1}$ are all causal by dual arguments.
\end{proof}

\begin{proposition}
$(\caus[\catc], \mathbf{<}, \mathbf{I})$ has a weak interchange with respect to $(\caus[\catc], \mathbf{\otimes}, \mathbf{I}, \mathbf{\parr}, \mathbf{I})$.
\end{proposition}

\begin{proof}
Since the units of these monoidal structures are all the same, $w_{I_\otimes} : \mathbf{I} \to \mathbf{I < I}$ is the trivial structural map $\lambda_I = \rho_I$ common to all of the structures. Coassociativity and commutativity with the unitors is immediate from coherence in $\catc$.

The actual structural map $w_\mathbf{\otimes} : \mathbf{(R < U) \otimes (T < V)} \to \mathbf{(R \otimes T) < (U \otimes V)}$ is taken to be $\id_R \otimes \sigma_{U,T} \otimes \id_V : (R \otimes U) \otimes (T \otimes V) \to (R \otimes T) \otimes (U \otimes V)$ (ignoring the usual coherence morphisms). Naturality and compatibility with the coherence morphisms is inherited from naturality and coherence in $\catc$. We just need to show that it is causal for this typing.

Given $h \in c_{\mathbf{R < U}} = c_\mathbf{R} \obslash c_\mathbf{U}$ and $k \in c_{\mathbf{T < V}} = c_\mathbf{T} \obslash c_\mathbf{V}$, suppose they have decompositions $\fsub{h} \sim \sum_i r_i \otimes \fsub{u_i}$ and $\fsub{k} \sim \sum_j t_j \otimes \fsub{v_j}$ where $\forall i . u_i \in c_\mathbf{U}$, $\exists r \in c_\mathbf{R} . \fsub{r} \sim \sum_i r_i$, $\forall j . v_j \in c_\mathbf{V}$, and $\exists t \in c_\mathbf{T} . \fsub{t} \sim \sum_i t_i$. Then $\fsub{(h \otimes k) \fatsemi (\id_R \otimes \sigma_{U,T} \otimes \id_V)} \sim \sum_{i,j} r_i \otimes t_j \otimes \fsub{u_i \otimes v_j}$. This is a state in $c_{\mathbf{(R \otimes T) < (U \otimes V)}} = c_{\mathbf{R \otimes T}} \obslash c_{\mathbf{U \otimes V}}$ since $\forall i, j . u_i \otimes v_j \in c_{\mathbf{U \otimes V}}$ and $\exists r \otimes t \in c_{\mathbf{R \otimes T}} . \fsub{r \otimes t} \sim \sum_{i,j} r_i \otimes t_j$.

The dual maps $w_{I_\parr} : \mathbf{I < I} \to \mathbf{I}$ and $w_\mathbf{\parr} : \mathbf{(C \parr E) < (D \parr F)} \to \mathbf{(C < D) \parr (E < F)}$ are given by the transposes of $w_{I_\otimes}$ and $w_\mathbf{\otimes}$. Compatibility between these and the linear distribution $\delta : \mathbf{A \otimes (B \parr C)} \to \mathbf{(A \otimes B) \parr C}$ (which is just the associator in $\catc$) is again given by coherence of $\catc$.
\end{proof}

\begin{theorem}[Restatement of Theorem \ref{thrm:bv_category}]\label{proof:thrm:bv_category}
$\caus[\catc]$ is a $\bv$-category.
\end{theorem}

\begin{proof}
$\caus[\catc]$ is a symmetric linearly distributive category with negation and a weak interchange. The units of $(\caus[\catc], \mathbf{\otimes}, \mathbf{I})$, $(\caus[\catc], \mathbf{<}, \mathbf{I})$, and $(\caus[\catc], \mathbf{\parr}, \mathbf{I})$ are all the same with the same unitors as morphisms in $\catc$ and the unit weak interchange maps are just unitors as well. The isomix map required to make it a $\bv$-category is just $1 : \mathbf{I} \to \mathbf{I}$.
\end{proof}

\begin{theorem}\label{proof:thrm:first_order_characterisation}[Restatement of Theorem \ref{thrm:first_order_characterisation}]
$\mathbf{A^* \parr A} = \mathbf{A^* < A} \Leftrightarrow \exists \mu \neq 0 . c_\mathbf{A}^* = \{\mu \cdot \discard_A\}$.
\end{theorem}

\begin{proof}
The cases for zero systems can be confirmed with Equations \ref{eq:zero_tensor} and \ref{eq:zero_one_way}. So suppose $\mathbf{A}$ is non-zero.

$\Rightarrow$: $\mathbf{A^* \parr A}$ will always contain the compact cup $\eta_A$ up to some invertible scalar. Suppose this is one-way signalling with left-residual $\pi \in c_{\mathbf{A^*}}$. Then by compact closure, any effect $\pi' \in c_\mathbf{A}^*$ must have $\pi' = \pi^*$. This means there is a unique effect, and by flatness this must be some multiple of the discard map $\discard_A$.

$\Leftarrow$: Given any $h : \mathbf{A^* \parr A}$ for some first-order $\mathbf{A}$, the one-way signalling condition holds trivially since there is a unique effect yielding a unique residual.
\end{proof}

\begin{proposition}\label{proof:prop:causality_exception}[Restatement of Proposition \ref{prop:causality_exception}]
$\mathbf{(I \oplus I) \parr A} = \mathbf{(I \oplus I) < A} \Leftrightarrow \mathbf{A} = \mathbf{1} \vee \exists \mu \neq 0 . c_\mathbf{A}^* = \{\mu \cdot \discard_A\}$
\end{proposition}

\begin{proof}
As usual, we can verify the zero system cases on their own, giving both $\mathbf{(I \oplus I) \parr 0} = \mathbf{(I \oplus I) < 0} \simeq \mathbf{0}$ and $\mathbf{(I \oplus I) \parr 1} = \mathbf{(I \oplus I) < 1} \simeq \mathbf{1}$. For the remainder, we suppose $\mathbf{A}$ is non-zero.

$\Leftarrow$: The one-way signalling condition holds trivially since there is a unique effect for $\mathbf{A}$, yielding a unique residual on $\mathbf{I \oplus I}$.

$\Rightarrow$: For the contrapositive, suppose that $\pi, \pi' \in c_\mathbf{A}^*$ are two distinct effects. By \ref{apc:basis} there must be some morphism $f \in \catc(I, A)$ that distinguishes them. Let $f + f' = \lambda \cdot \maxmix_A$ via Proposition \ref{prop:bintests}, and let $\mu$ be the invertible scalar such that $\mu \cdot \maxmix_A \in c_\mathbf{A}$ by flatness.

We construct $h := \mu \cdot \lambda^{-1} \cdot (\iota_1 \otimes f + \iota_2 \otimes f') \in \catc(I, (I \oplus I) \otimes A)$. This is of type $\mathbf{(I \oplus I) \parr A}$ since $h \fatsemi (\discard_{I \oplus I} \otimes \id_A) = \mu \cdot \lambda^{-1} \cdot (f + f') = \mu \cdot \maxmix_A \in c_\mathbf{A}$. However, $h$ is not in $c_{\mathbf{(I \oplus I) < A}}$ since $h \fatsemi (\id_{I \oplus I} \otimes \pi) \neq h \fatsemi (\id_{I \oplus I} \otimes \pi')$ because projecting out the first element would give unequal scalars $\mu \cdot \lambda^{-1} \cdot (f \fatsemi \pi) \neq \mu \cdot \lambda^{-1} \cdot (f \fatsemi \pi')$.
\end{proof}

\section{Proofs for Section \ref{sec:nonsignalling}}

\begin{theorem}[Restatement of Theorem \ref{thrm:tensor_nonsignalling}]\label{proof:thrm:tensor_nonsignalling}
$(c_\mathbf{A} < c_\mathbf{B}) \cap (c_\mathbf{A} > c_\mathbf{B}) = c_{\mathbf{A \otimes B}}$.
\end{theorem}

\begin{proof}
We refer to Equations \ref{eq:zero_tensor} and \ref{eq:zero_one_way} for the cases involving zero systems, allowing us to assume we are working with non-zero systems for the remainder of this proof.

The $\supseteq$ direction is easy since any affine combination of product states $\fsub{h} \sim \sum_i \alpha_i \cdot \fsub{f_i \otimes g_i}$ will be non-signalling with left-residual characterised by the decomposition $\fsub{h \fatsemi (\id_A \otimes \pi)} \sim \sum_i \alpha_i \cdot \fsub{f_i}$ for any choice of $\pi \in c_\mathbf{B}^*$ and a similarly constructed right-residual. The remainder of the proof concerns the $\subseteq$ direction, i.e. showing that all non-signalling processes are affine combinations of product processes.

Consider any $h \in (c_\mathbf{A} < c_\mathbf{B}) \cap (c_\mathbf{A} > c_\mathbf{B})$ with left-residual $m \in c_\mathbf{A}$ and right-residual $n \in c_\mathbf{B}$. Since $h \in c_\mathbf{A} \obslash c_\mathbf{B} = c_\mathbf{A} < c_\mathbf{B}$ by Lemma \ref{lemma:one_way_equals_pcs}, we have some $\{f_i\}_i, \{g_i\}_i$ such that $\fsub{h} \sim \sum_i f_i \otimes \fsub{g_i}$, $\fsub{m} \sim \sum_i f_i$ with $\forall i . g_i \in c_\mathbf{B}$. Take some maximal linearly-independent subset of $\{g_i\}_i$ and extend it to a basis $\{\rho_j\}_j$ for $B$ in $\catc$ with a dual basis $\{e_j\}_j$ in $\sub(\catc)$ using Lemma \ref{lemma:extend_basis}. Similar to the proof of Lemma \ref{lemma:one_way_equals_pcs}, we ensure that $\fsub{n} \fatsemi e_j \neq 0$ (and hence has an inverse) for every $\rho_j$ corresponding to some linear combination of $\left\{g_i\right\}_i$ - we will denote this subset of the indices as $\mathcal{I}$. Since this basis was formed by a maximal linearly-independent set, every $g_i$ must be a linear combination of only the $\mathcal{I}$ terms and hence $\forall j \notin \mathcal{I} . \fsub{g_i} \fatsemi e_j \sim 0$.

For any $j \notin \mathcal{I}$ and any effect $\pi \in c_\mathbf{A}^*$, we have:
\begin{equation}
\begin{split}
\fsub{h} \fatsemi (\id_A \otimes e_j) &\sim \sum_i (f_i \otimes \fsub{g_i}) \fatsemi (\id_A \otimes e_j) \\
&\sim \sum_i (\fsub{g_i} \fatsemi e_j) \cdot f_i \\
&\sim \sum_i \fsub{0_{I,I}} \cdot f_i \\
&\sim \fsub{0_{I,A}}
\end{split}
\end{equation}

\begin{equation}
\begin{split}
\fsub{n} \fatsemi e_j &\sim \fsub{h} \fatsemi (\fsub{\pi} \otimes e_j) \\
&\sim \fsub{0_{I,A}} \fatsemi \fsub{\pi} \\
&\sim \fsub{0_{I,I}}
\end{split}
\end{equation}

By resolving the identity on $B$, we have:
\begin{equation}
\begin{split}
\fsub{h} &\sim \sum_j (\fsub{h} \fatsemi (\id_A \otimes e_j)) \otimes \fsub{\rho_j} \\
&\sim \sum_{j \in \mathcal{I}} (\fsub{h} \fatsemi (\id_A \otimes e_j)) \otimes \fsub{\rho_j} \\
&\sim \sum_{j \in \mathcal{I}} (\fsub{n} \fatsemi e_j)^{-1} \cdot (\fsub{h} \fatsemi (\id_A \otimes e_j)) \otimes (\fsub{n} \fatsemi e_j) \cdot \fsub{\rho_j}
\end{split}
\end{equation}

\begin{equation}
\begin{split}
\fsub{n} &\sim \sum_j (\fsub{n} \fatsemi e_j) \cdot \fsub{\rho_j} \\
&\sim \sum_{j \in \mathcal{I}} (\fsub{n} \fatsemi e_j) \cdot \fsub{\rho_j}
\end{split}
\end{equation}

Analogously to the proof of Lemma \ref{lemma:one_way_equals_pcs}, we can define $f_j := (\fsub{n} \fatsemi e_j)^{-1} \cdot (\fsub{h} \fatsemi (\id_A \otimes e_j))$ which will be normalised by any $\pi \in c_\mathbf{A}^*$, then use Proposition \ref{prop:bintests} to find some $f_j' \in \catc(I,A)$ such that $\fsub{f_j'} \sim \fsub{\lambda_j \cdot \maxmix_A} + f_j$ and therefore $\alpha \cdot f_j' \in c_\mathbf{A}$ for some invertible scalar $\alpha \in \catc(I,I)$. We follow in the same manner until we reach the following equation:
\begin{equation}
\begin{split}
\fsub{h} &\sim \sum_{j \in \mathcal{I}} \fsub{f_j'} \otimes (\fsub{n} \fatsemi e_j) \cdot \fsub{\rho_j} - \sum_{j \in \mathcal{I}} \fsub{\lambda_i \cdot \maxmix_A} \otimes (\fsub{n} \fatsemi e_j) \cdot \fsub{\rho_j}
\end{split}
\end{equation}

Since each of $\{\rho_j\}_{j \in \mathcal{I}}$ is a linear combination of $\{g_i\}_i$, this represents $h$ as a linear combination of terms in $c_\mathbf{A} \otimes c_\mathbf{B}$. The normalisation of $h \in c_{\mathbf{A \parr B}}$ guarantees that this linear combination is affine, and hence $h \in c_{\mathbf{A \otimes B}} = (c_\mathbf{A} \otimes c_\mathbf{B})^{**}$ by Theorem \ref{thrm:affine}.
\end{proof}

\end{document}